%% file: Manuscript.tex
\documentclass[journal]{IEEEtran}

\input{macros}
\usepackage{cuted}
\AfterEndEnvironment{strip}{\leavevmode}
\usepackage{amsmath,amssymb,amsfonts,latexsym,textcomp,nccmath,mathtools}
\usepackage{microtype}
\usepackage{mleftright}
\usepackage{multicol}
\mleftright % cut down the whitespace added to each \left-\right pair of round, square, and curly fences.
\usepackage{easylist}

\begin{document}
	 \title{Decentralized Stealth Attacks on Cyber-Physical Systems}
	 \author{Xiuzhen~Ye,
	 	I\~naki Esnaola, 
	 	Samir M. Perlaza,
	 	and Robert~F. Harrison
	 	     \begin{multicols}{1}
	 	     	\begin{@twocolumntrue}  
	 	\thanks{
   This research was supported jointly by the Department of Automatic Control and Systems Engineering, University of Sheffield and China Scholarship Council.
   
	 		X. Ye was with the Department of Automatic Control and Systems Engineering, University of Sheffield, Sheffield S1 3JD, UK. I. Esnaola and R.~F. Harrison are with the Department of Automatic Control and Systems Engineering, University of Sheffield, Sheffield S1 3JD, UK.
	 		I. Esnaola is also with the Department of Electrical Engineering, Princeton University, Princeton NJ
	 		08544, USA. (email: xzhen0921@gmail.com, esnaola@sheffield.ac.uk, r.f.harrison@sheffield.ac.uk).
	 		
	 		S.~M. Perlaza is with the Institut National de Recherche en Informatique
	 		et Automatique (INRIA), Lyon, France, and also with the Department of
	 		Electrical Engineering, Princeton University, Princeton, NJ 08544 USA (email: samir.perlaza@inria.fr).}
 		\end{@twocolumntrue}
 	\end{multicols}
	 }
	 
\maketitle
\begin{abstract}	
	Decentralized stealth attack constructions that minimize the mutual information between the state variables and the measurements are proposed. The attack constructions are formulated as random Gaussian attacks targeting Cyber-physical systems that aims at minimizing the mutual information between the state variables and measurements while constraining the Kullback-Leibler divergence between the distribution of the measurements under attacks and the distribution of the measurements without attacks. The proposed information metrics adopted measure the disruption and attack detection both globally and locally. The decentralized attack constructions are formulated in a framework of normal games. The global and local information metrics yield games with global and local objectives in disruption and attack detection. We have proven the games are potential games and the convexity of the potential functions followed by the uniqueness and the achievability of the Nash Equilibrium, accordingly. We proposed a best response dynamics to achieve the Nash Equilibrium of the games. We numerically evaluate the performance of the proposed decentralized stealth random attacks on IEEE test systems and show it is feasible to exploit game theoretic techniques in decentralized attack constructions.
\end{abstract}

\section{Introduction}

%\textcolor{red}{We need to agree on temrs: observation vs measurement; Attack vector vs. attack; CPS vs system}

Cyber-physical systems (CPSs) are complex engineered systems that emerge from integrating a cyber layer with physical systems consisting of processes, devices, and infrastructure. There is a growing interest in CPSs from both the public and private sectors. The US National Science Foundation CPSs as a key research area \cite{NSF_web}, recognizing their potential for improved efficiency and functional performance across different sectors, while the UK government recognizes them as a key to embed resilience into infrastructure \cite{Advisor_UK_24}. 
The applications of CPSs are vast and varied, spanning from local systems like building energy management \cite{AR_TonEmergingTopicsinComputing_17} to large-scale integrations such as city-wide autonomous transportation systems \cite{HZ_TVT_15}. They form the basis for smart grids, industrial control systems (ICSs) \cite{SK_ProceedingsofIEEE_13}, manufacturing automation in Industry 4.0 \cite{MK_CirpAnnals_16,LF_Industry4_14}, and biomedical and healthcare systems \cite{ZQ_SystemsJournal_15}. A representative example is the supervisory control and data acquisition (SCADA) system, widely used in smart grids, manufacturing industries, and telecom and IT-based systems \cite{HR_IoT_17}.

CPSs are characterized by a distributed and dynamic interaction between computational and physical components \cite{RH_ICT_11}, often comprising a multitude of spatial, temporal, and functional scales. To enable this interface, they incorporate operational technology (OT) in the form of smart sensing and embedded computation, operating with heterogeneous communication systems ranging from edge to cloud computing. The enabler of this distributed new paradigm is the multiscale flow of data across all system stakeholders, which opens the door to data integrity threats on the communication networks that facilitate the exchange of operational information between elements of the system. This increased integration also creates increased risks due to the evolution of the cyber layer and the vulnerabilities that it introduces into the system.

The reliance on OT for integration of the cyber and physical layers exposes systems to various attacks, including data injection attacks (DIAs) \cite{LY_TISSEC_11}, denial of service (DoS), and breaches of customer information \cite{HL_IEEEIoT_17}. These risks are amplified by the increased complexity of CPSs and the distributed nature of CPSs which facilitates decentralized attack constructions that require minimum coordination between the attackers. 
To address these challenges, researchers from various disciplines are contributing to the development of secure CPSs \cite{Advisor_US_24}. Key approaches include control technology to maintain secure and economic operation of CPSs by reacting rapidly and autonomously to cyber threats and the use of AI to detect attacks and mitigate their impact. 

In this setting, information theory offers quantitative measures of the information in the data flows within CPSs, establishing in this manner a mathematical framework to study the impact of data integrity on CPSs. It also enables mathematical descriptions of the fundamental detection and estimation limits, which combined paves the way for innovative approaches to cybersecurity and reliability \cite{ZJ_WonSE_15}.
As CPSs continue to evolve, providing resilience and security guarantees will be crucial for realizing their full potential across numerous applications and industrial sectors \cite{AA_book_20}.
This paper is concerned with modelling and analyzing one of the major data integrity threats to CPSs: DIAs on the state estimation processes of CPSs. 
In particular, we aim to characterize the fundamental limits of the threats and vulnerabilities that decentralized attack strategies pose on CPSs and the effect that different system parameters, such as sensing infrastructure and network topology, have on the overall state estimation processes of the system.

DIAs alter the state estimate of the systems by compromising the data acquisition and communication procedures in CPS without triggering bad data detection mechanisms~\cite{LY_TISSEC_11}. The interaction between the different sub-systems that comprise a CPS poses the threat of decentralized DIAs where the attackers construct DIAs locally wiht the objective of disrupting the state estimate of the system~\cite{IE_TSG_16}. Decentralized DIAs are studied in~\cite{IE_TSG_16,CKKPT_SPM_12,EPP_gsip_14,KP_TSG_11} under the assumption that bad data detection is performed via a residual test and that the state estimation is least squares estimation. Naturally, decentralized DIAs with attackers at different locations lead to the exploration of game theoretic notions~\cite{OM_MIT_94} to devise attack construction strategies. A comprehensive survey on security control and attack detection for CPS is available in~\cite{EB_dynamicGames_19}.
A game-theoretic approach to model and quantify the security of CPS is proposed in~\cite{OA_ComInd_17}. 
The maximum distortion MMSE attacks to cyberphysical energy systems are studied analytically and numerically in~\cite{IE_TSG_16} for both centralized and decentralized settings. 
Therein, it is shown that when the attack disruption is measured in terms of the quadratic error of the state estimation, game theoretic techniques describe the coordination requirements for optimal decentralized attack constructions.

The unprecedented data acquisition capabilities in CPSs facilitates access to large historical data sets that can in turn be distilled into stochastic models of the system. This modelling paradigm is leveraged by system operators for optimized management and control, but also by the attackers. In particular, it allows attackers to formulate DIAs that exploit the statistical structure of the CPS to maximize the disruption and minimize the probability of detection. Moreover, access to historical data implies that data-driven attack constructions~\cite{MD_TSG_22}~\cite{TW_CS_22} are feasible, and therefore, predicting specific attack strategies is challenging in practical terms. 
For that reason, it is essential to study DIAs on CPS in terms of the fundamental limits governing the state estimation procedures and data integrity attacks that encompass a wide range of data availability paradigms. For instance, in~\cite{Ali_SAandMSP_14}, a framework for the analysis of state estimate under structured data is studied. State variables are modelled as a multivariate Gaussian distribution within the Bayesian framework whose second order moments are available to the attackers to construct random attacks in~\cite{OK_TSG_11, SE_SGC_17, SE_TSG_19, YE_SGC_20}. 
Within the Bayesian framework, the attack detection is formulated as the likelihood ratio test~\cite{JN_LRT_33} or alternatively machine learning techniques~\cite{OM_TNNLS_16}. With the exploitation of statistical structure of the CPS, the DIAs are cast within a Bayesian framework that allows information theoretic measures to characterize the corresponding fundamental limits~\cite{TM_ElementsofIT}.
DIAs with information theoretic measures were first studied in~\cite{SE_SGC_17} and generalized in~\cite{SE_TSG_19}. Sparsity constraints are considered in the same setting for independent attacks in~\cite{YE_SGC_20} and generalized in~\cite{YE_TSG_21}. However, the studies in~\cite{SE_SGC_17, SE_TSG_19, YE_SGC_20,YE_TSG_21} are for a single centralized attack setting.
Given the complexity and heterogeneity of CPSs, the characterization of threats and vulnerabilities is challenging in general terms. For that reason, studying DIAs from a fundamental standpoint can provide valuable operational insight on the analysis of the vulnerability and design of corresponding protection mechanisms for CPSs. 

The main contribution of this work is the characterization and analysis of decentralized attack constructions using information theoretic disruption measures in a Bayesian setting. 
The attackers are fully distributed, and thus, game theoretic tools are adopted to determine the optimal local action strategy to achieve their aim. The cost of the attacks is formulated in terms of the information theoretic description of the disruption, that is, the mutual information between the state variables and the observations; and the Kullback-Leibler (KL) divergence between the normal operation observations and the attacked observations. 
This research provides a game result of the information theoretic attacks in decentralized CPS where the optimal behaviour of the attackers are studied analytically and numerically. Specifically, the existence and uniqueness of the Nash Equilibria (NE) of the game are provided.
The main contributions of this paper are:
(1) A novel decentralized information theoretic stealth attack construction is proposed. 
(2) Information measures are used to design cost functions for decentralized random attacks constructions that jointly minimize mutual information and KL divergence. 
(3) The disruption of local and global attack strategies are compared in information theoretic terms, as well as the stealth guarantees that they achieve.
(4) The decentralized attack construction is formulated as a potential game and the resulting Nash Equilibrium (NE) is analyzed.
(5) Three different potential games that emerge from assuming different degrees of attacker knowledge and coordination. The results show that the potential functions are convex and that the NE is unique in all the proposed games. The insight obtained from the properties of the games are distilled to propose best response dynamics that yield practical attack constructions. 

The rest of the paper is organized as follows: In Section~\ref{sec_system_model}, we introduce DIAs to CPS with linearized dynamics within a Bayesian framework. Information theoretic metrics are presented in Section \ref{sec_Information_Theoretic_Metrics}. The game formulation and the Nash Equilibrium are introduced in Section~\ref{sec_game_formulation} and Section~\ref{sec_NE}, respectively. In Section~\ref{sec_numerical_results}, we evaluate the performance of the proposed attack constructions on a representative cyber physical system. The paper closes with conclusions in Section~\ref{sec_conclusion}.

\textbf{Notation:} 
The elementary vector $\ev_i\in\mathds{R}^n$ is a vector of zeros with a one in the $i$-th entry. 
Random variables are denoted by capital letters and their realizations by the corresponding lower case, e.g. $x$ is a realization of the random variable $X$. 
Vectors of $n$ random variables are denoted by a superscript, e.g. $X^n=(X_1,X_2, \ldots, X_n)^{\sf{T}}$ with corresponding realizations denoted by $\xv$. 
The set of positive semidefinite matrices of size $n\times n$ is denoted by $S_{+}^n$. The $n$-dimensional identity matrix is denoted as $\textbf{I}_n$.
Given an $n$-dimensional vector $\muv \in \mathds{R}^n$ and a matrix $\Sigmam \in S_{+}^n$, we denote by $\mathcal{N} (\muv, \Sigmam)$ the multivariate Gaussian distribution of dimension $n$ with mean $\muv$ and covariance matrix $\Sigmam$.
The mutual information between random variables $X$ and $Y$ is denoted by $I(X;Y)$ and the Kullback-Leibler (KL) divergence between the distributions $P$ and $Q$  is denoted by $D(P\| Q)$.
A set is denoted by $\Vc$ and $\Vc^m$ denotes the Cartesian product of $m$ sets $\Vc$.
We denote the number of state variables on a given system by $n$ and the number of the observations by $m$.
We denote the $i$-th entry of vector $X^n$ by $(X^n)_i$ and the Minkowski sum of $X^n$ and $Y^n$ by $X^n \otimes Y^n$.
\section{Cyber-Physical Systems}\label{sec_system_model}

\subsection{Observation Model}

Let $\xv \in{\mathds{R}^n}$ be a vector containing the state variables that describe the state of the CPS. The state variables are observed through the acquisition function $F: {\mathds{R}^n} \rightarrow {\mathds{R}^m}$ that is determined by the system components and topology of the system. A linearized observation model for state estimation corrupted by noise is given by
\be\label{eq:obs_noattack}
Y^m  = \Hm\xv+Z^m,
\ee
where $\Hm \in {\mathds{R}^{m \times n}}$ is the Jacobian of the function $F$ at a given operating point. The observation vector $Y^m \eqdef \left(Y_1,Y_2,\ldots,Y_m\right)$ is corrupted by additive white Gaussian noise (AWGN), c.f.,~\cite{GJ_PSanalysis_1994,AA_PSstateestimation_04}. The noise represented by the random vector $Z^m \eqdef \left(Z_1,Z_2,\ldots,Z_m\right)$ in \eqref{eq:obs_noattack}, is modelled as a multivariate Gaussian distribution, that is,
\be 
\label{EqZ}
Z^m \sim \mathcal{N}(\textbf{0},\sigma^2 \textrm{\textbf{I}}_m),
\ee
where $\textbf{0} = \left( 0,0, \ldots, 0 \right)^{\sf T} \in \mathds{R}^m$ and $\sigma^2 \textrm{\textbf{I}}_m$ are the mean vector and the covariance matrix of the random additive noise.

We adopt a Bayesian estimation framework, and therefore, the state variables are denoted by a random vector $X^n ~\eqdef~\left(X_1,X_2,\ldots,X_n\right)^{\sf T}$. In this study, we assume that the random vector $X^n$ follows a multivariate Gaussian distribution with a zero mean vector and covariance matrix $\Sigmam_{X\!X} \in S_{+}^n$, that is,
\be\label{Sigma_XX}
X^n \sim \mathcal{N}(\textbf{0}, \Sigmam_{X\!X}).
\ee

The rationale for modelling the state variables as given in~\eqref{Sigma_XX} stems from both first principles modelling and data considerations~\cite{SE_TAC_17,GE_spawc_16}. It is important to note that, additionally, the Gaussian assumption induces the maximum entropy distribution over the state variables. In doing so, we assign to the state variables the distribution that minimizes the amount of inductive bias introduced by the random modelling framework.
Consequently, the  vector of measurements $Y^m$ in~\eqref{eq:obs_noattack} follows a multivariate Gaussian distribution $P_{Y^m}$ with a null mean vector and covariance matrix $\Sigmam_{Y\!Y}$, that is,
\be\label{eq:obs_no_attack_distribution}  
 Y^m \sim P_{Y^m} = \mathcal{N}(\textbf{0}, \Sigmam_{Y\!Y}),
\ee
with
\be\label{20220830_4}
\Sigmam_{Y\!Y}\eqdef\Hm \Sigmam_{X\!X}\Hm^{\sf T}+\sigma^2 \textrm{\textbf{I}}_m.
\ee 
Note that for all $i\in\{1,2,\ldots,m\}$, the probability distribution of the $i$-th entry of the random vector $Y^m$ is denoted by $P_{Y_i}$, that is,
\be\label{ith_noatt}
Y_i \sim P_{Y_i} = \mathcal{N}(0, \ev_i^{\sf T}\Sigmam_{Y\!Y}\ev_i).
\ee

\subsection{Decentralized Data Injection Attacks}\label{sec_Decentralized Data Injection Attacks}

The attacker manipulation of the measurements in~\eqref{eq:obs_noattack} can be cast as additive FDIAs~\cite{LY_TISSEC_11}. Within a Bayesian framework in~\eqref{eq:obs_no_attack_distribution}, given the stochastic nature of the measurements, the attacker pursues a random attack construction strategy with the aim of achieving an performance described in terms of the expectation of the resulting joint distribution for the state variables, the observations, and the attack variables.
A random malicious attack, denoted by a random vector $A^m$, compromises the integrity of the vector of measurements in~\eqref{eq:obs_no_attack_distribution}, which yields the vector of compromised measurements as follows: 
\be 
	\label{eq:obs_attack}
	Y_A^m  = \Hm X^n+Z^m + A^m,
\ee
where $A^m \sim P_{A^m}$ and $P_{A^m}$ is the distribution of the random attack vector $A^m$. In this study, $P_{A^m}$ is assumed to be a multivariate Gaussian distribution that satisfies
\be\label{eq:Gauss_attack}
  A^m \sim P_{A^m} = \mathcal{N} (\textbf{0}, \Sigmam_{A\!A}),
\ee
where $\Sigmam_{A\!A}\in S_{+}^m$ is a covariance matrix. 
The argument of the Gaussian assumption in~\eqref{eq:Gauss_attack} is the fact that given a fixed covariance matrix  $\Sigmam_{A\!A}$, a multivariate Gaussian distribution minimizes the mutual information between the state variables $X^n$ and the compromised measurements $Y_A^m$ in~\eqref{eq:obs_attack}~\cite{SI_TIT_13}. The zero mean assumption for $A^m$ follows from the fact
that the mutual information does not depend on the mean of $A^m$~\cite{InriaRR9481}, and therefore, there is no loss of generality.

 Stealth attacks with a unique attacker, referred to as {\it centralized attacks}, are studied in~\cite{SE_TSG_19}. In a decentralized setting, DIAs are constructed by several attackers that have access to the measurements in the system and are referred to as {\it decentralized attacks}~\cite{IE_TSG_16}. In this setting, attackers construct the attack vector $A \in \mathds{R}^m$ with distributively. The aim of one attackers is to independently agree on the attack vector that maximizes the disruption to the system, i.e. information loss induced by the attack, while staying undetected. 
 %
 %All the attackers have share same goal, which reveals a cooperative manner among the attackers in  decentralized attacks.
%
Therefore, the resulting attack $A^m$ in~\eqref{eq:Gauss_attack} is modelled as independent across the entries of attack vector, that is, 
\be\label{A_i_def}
A^m \eqdef (A_1, A_2, \ldots, A_m)^{\sf T},
\ee
such that
\be\label{att_independency}
P_{A^m} = \prod_{i=1}^mP_{A_i}
\ee
where, for all $i \in \{1,2,\ldots,m\}$, the probability density function of $P_{A_i}$ is Gaussian with zero mean and variance $v_i \in \left[0,+\infty\right)$, that is,
\be\label{distribution_A_i}
A_i \sim {\cal N} (0, v_i).
\ee
The independence between the entries of the random attack in~\eqref{att_independency} implies that attackers do not require communication between different attack locations. That being the case, the decentralized attack construction is feasible without any coordination requirements, and therefore, poses a particularly interesting threat from a practical standpoint.

Note that with the independence, the covariance matrix in~\eqref{eq:Gauss_attack} is
\be\label{20220519_6}
\Sigmam_{A\!A} = \sum_{i=1}^m v_i \ev_i\ev_i^{\sf T},
\ee 
and as a result the vector of compromised measurements $Y^m_A$ under the random attacks follows a multivariate Gaussian distribution $P_{Y_A^m}$ given by
\begin{equation}\label{eq:obs_attack_distribution} 
  Y^m_A  \sim P_{Y_A^m} = \mathcal{N} (\textbf{0}, \Sigmam_{Y_A\!Y_A})
\end{equation}
with 
\be
\Sigmam_{Y_A\!Y_A} \eqdef \textbf{H}\Sigmam_{X\!X}\textbf{H}^{\sf{T}} + \sigma^2 \textrm{\textbf{I}}_m +  \Sigmam_{A\!A}. 
\ee
From~\eqref{ith_noatt} and~\eqref{distribution_A_i}, for all $i \in \{1,2,\ldots,m\}$, 
the $i$-th entry of the random vector $Y_A^m$ denoted by $Y_{A,i}$ is such that 
\be\label{ith_compromised}
Y_{A_i} \sim P_{Y_{A,i}} = \mathcal{N}(0, \ev_i^{\sf T}\Sigmam_{Y\!Y}\ev_i + v_i),
\ee
where $\Sigmam_{Y\!Y}$ is defined in~\eqref{20220830_4} and $v_i$ is introduced in~\eqref{distribution_A_i}.

The independence assumption of the entries of random attack vector in~\eqref{att_independency} allows the attackers to launch attacks independently, i.e. without any shared randomness. 
Let ${\cal M} \eqdef \{1,2,\ldots,m\}$ be the set of measurement indices in the system and
\be\label{set_K}
\Kc \eqdef \{1,2,\ldots,m\}
\ee
be the set of indices of attacked measurements. Assume that measurement $i$, with $i \in \Mc$, is the only measurement that attacker $i$ can compromise. Let $A_i^m \in \mathds{R}^m$ be the random attack vector produced by attacker $i$ and ${\cal A}_i$ be the set of random attack vectors that can be injected into the system by attacker $i$, with $ i \in {\cal K}$, that is,
\be\label{20220920_1}
\Ac_i = \{A_i^m \in \mathds{R}^m: (A^m_i)_j = 0 \ \textnormal{for all} \ j \neq i \}.
\ee
Hence, from~\eqref{A_i_def} and~\eqref{20220920_1}, for all $i \in \Kc$, the following holds
\be\label{ith_vector}
A^m_i = A_i \otimes \ev_i,
\ee
where $A_i \otimes \ev_i$ is the Kronecker product of $A_i$ and $\ev_i$.
Let the Minkowski sum of ${\cal A}_i$ and ${\cal A}_j$ be denoted by ${\cal A}_i \oplus {\cal A}_j$. For all $A^m \in {\cal A}_i \oplus {\cal A}_j$, there exists a pair of random vectors $(A_i^m, A_j^m) \in {\cal A}_i \times {\cal A}_j$ such that $A^m = A_i^m+A_j^m$. Let the set of all possible random attack vectors be
\be
{\cal A} \eqdef {\cal A}_1 \oplus {\cal A}_2 \oplus \ldots \oplus {\cal A}_m,   
\ee
and the set of complementary random attack vectors with respect to the attacker $i$ be
\be
{\cal A}_{-i} \eqdef {\cal A}_1 \oplus {\cal A}_2 \oplus \ldots {\cal A}_{i-1} \oplus {\cal A}_{i+1} \ldots \oplus {\cal A}_{m-1} \oplus {\cal A}_m.
\ee
Denote the random attack vector constructed by the attacker $i$ by $A_i^m \in {\cal A}_i$. Hence, the resulting random attack vector is $A^m \in \mathds{R}^m$ in~\eqref{eq:obs_attack} and satisfies
\be\label{20220518_2}
A^m = \sum_{i \in \Kc} A_i^m \in {\cal A}.
\ee
Denote also the complementary random attack vector of $A_i^m$ as follows:
\be
A^m_{-i} \eqdef \sum_{j \in {\cal K} \setminus \{i\}} A^m_j \in {\cal A}_{-i}.
\ee

Given the actions by all the other attackers $A^m_{-i}$, the aim of attacker $i$ is to corrupt the measurements by injecting its random attack vector $A_i^m \in {\cal A}_i$ to compromise the data integrity while guaranteeing a low probability of attack detection. For modelling this behaviour, attacker~$i$, with $i \in \Kc$, adopts the cost function $\phi_i$: $\mathds{R}^m \rightarrow \mathds{R}$ to determine whether a random attack vector $A^m_i \in {\cal A}_i$ is more beneficial than another attack vector $B^m_i \in {\cal A}_i$.
In this context, the attack vector $A^m_i$ is preferred to $B^m_i$ if~$\phi_i (A_i^m + A_{-i}^m) < \phi_i (B_i^m + A_{-i}^m)$.  

 \subsection{Attack Detection}
 As a part of security strategies, the system operator implements an attack detection procedure prior to performing state estimation. The detection is cast as a hypothesis test with hypotheses
 \begin{subequations}\label{EqHypTestA}
 	\begin{IEEEeqnarray}{rll} 
 		\mathcal{H}_0&:\textrm{There is no attack,}\\  
 	\mathcal{H}_1&:\textrm{Measurements are compromised}.
 \end{IEEEeqnarray}
\end{subequations}
The following detection frameworks are studied in this paper.
\subsubsection{Joint Detection}\label{sec_joint_detection}
The operator decides if a vector of observations is produced under attack. Specifically, the operator acquires a vector of measurements $\bar{Y}^m$ and decides whether it is produced under attack. 
In a Bayesian framework, the hypothesis test problem in~\eqref{EqHypTestA} becomes 
\begin{subequations}\label{eq:hypoth_attack}
\begin{IEEEeqnarray}{rll}
  \mathcal{H}_0&: \bar{Y}^m\thicksim\Nc(\mathbf{0},\Sigmam_{Y\!Y}), \\
  \mathcal{H}_1&: \bar{Y}^m\thicksim\Nc(\mathbf{0},\Sigmam_{Y_A\!Y_A}).
\end{IEEEeqnarray}
\end{subequations}
A deterministic test $T_{\textnormal{JD}}:\mathds{R}^m\rightarrow\{0,1\}$ is adopted to decide which distribution generates the measurements. Given a measurement vector $\mathbf{\bar{y}}$, let $T_{\textnormal{JD}}(\mathbf{\bar{y}} )=i$ denote the case in which the test decides for $\Hc_i$ upon $\mathbf{\bar{y}}$, with $i \in \{0,1\}$. Therefore, the deterministic test $T_{\textnormal{JD}}$ is
\begin{subequations} 
	\begin{IEEEeqnarray}{rll}
		\mathcal{H}_0&: T_{\textnormal{JD}}(\mathbf{\bar{y}} )=0, \\
		\mathcal{H}_1&: T_{\textnormal{JD}}(\mathbf{\bar{y}} )=1.
	\end{IEEEeqnarray}
\end{subequations}
The performance of the test is assessed in terms of the Type-I error, denoted by $\alpha\eqdef \mathds{P}\left[T_{\textnormal{JD}}\left( \bar{Y}^m \right)=1\right]$, with $\bar{Y}^m$ in \eqref{eq:obs_no_attack_distribution}; 
and the Type-II error, denoted by $\beta\eqdef \mathds{P}\left[T_{\textnormal{JD}}\left(\bar{Y}^m\right)=0 \right]$, with $\bar{Y}^m$ in \eqref{eq:obs_attack_distribution}. 
Given the Type-I error satisfies $\alpha \leq  \alpha'$, with $\alpha'\in[0,1]$, the likelihood ratio test (LRT) is optimal in the sense that it induces the smallest Type-II error $\beta$~\cite{JN_LRT_33}. For the given measurement vector $\bar{\yv}$, the LRT is given by
\begin{equation}\label{lrt}
	T_{\textnormal{JD}}(\mathbf{\bar{y}}) = \mathds{1}_{\left\lbrace L_{\textnormal{JD}}(\mathbf{\bar{y}}) \geqslant \tau \right\rbrace},
\end{equation}
with $\tau\in\mathds{R}_+$ is the decision threshold that meets the constraint on Type-I error set by the operator and $L_{\textnormal{JD}}(\mathbf{\bar{y}})$ the likelihood ratio given by
\begin{equation} 
L_{\textnormal{JD}}(\mathbf{\bar{y}}) = \frac{f_{Y_A^m}(\mathbf{\bar{y}})}{f_{Y^m}(\mathbf{\bar{y}})},
\end{equation}
where the functions $f_{Y_A^m}$ and $f_{Y^m}$ are the probability density function of $Y_A^m$ in~\eqref{eq:obs_attack} and $Y^m$ in~\eqref{eq:obs_noattack}, respectively.
Note that changing the value of $\tau$ is equivalent to changing the tradeoff between Type-I and Type-II errors.

\subsubsection{Local Detection}\label{sec_local_detection}
The operator decides if the local measurement $i$, with $i\in \Mc$, is produced under attack. Note that in this case, instead of observing the vector of all measurements, the operator aims to  establish if the particular measurement under consideration is attacked based on the local information available without giving the rest of the system measurements any consideration.
Specifically, the operator acquires measurement $i$ denoted by $\bar{Y}_i$ and decides whether it is produced under attack. 
Within the Bayesian framework, the hypothesis test problem in~\eqref{EqHypTestA} becomes 
\begin{subequations} 
	\begin{IEEEeqnarray}{rll}
		\mathcal{H}_0&: \bar{Y}_i \thicksim \mathcal{N}(0, \ev_i^{\sf T}\Sigmam_{Y\!Y}\ev_i), \\
		\mathcal{H}_1&: \bar{Y}_i \thicksim \mathcal{N}(0, \ev_i^{\sf T}\Sigmam_{Y\!Y}\ev_i + v_i).
	\end{IEEEeqnarray}
\end{subequations}
A deterministic test $T_{\textnormal{ID}}:\mathds{R} \rightarrow\{0,1\}$ is adopted to determine which distribution generates the measurements. Given measurement $i$ denoted by $ \bar{y}_i$, let $T_{\textnormal{ID}}( \bar{y}_i)=j$ denote the case in which the test decides for $\Hc_j$ upon $\bar{y}_i$, with $j \in \{0,1\}$. Therefore, the deterministic test $T_{\textnormal{ID}}$ is
\begin{subequations} 
	\begin{IEEEeqnarray}{rll}
		\mathcal{H}_0&: T_{\textnormal{ID}}( \bar{y}_i )=0, \\
		\mathcal{H}_1&: T_{\textnormal{ID}}( \bar{y}_i )=1.
	\end{IEEEeqnarray}
\end{subequations}
In this setting, the LRT is given by
\begin{equation}\label{lrt_local}
	T_{\textnormal{ID}}(\bar{y}_i) = \mathds{1}_{\left\lbrace L_{\textnormal{ID}}(\bar{y}_i) \geqslant \tau \right\rbrace},
\end{equation}
with $\tau \in \mathds{R}_+$ the decision threshold and the likelihood ratio $L_{\textnormal{ID}}(\bar{y}_i)$ given by
\begin{equation} 
	L_{\textnormal{ID}}(\bar{y}_i) = \frac{f_{Y_{A,i}}( \bar{y}_i)}{f_{Y_{A,i}}( \bar{y}_i)},
\end{equation}
where the functions $f_{Y_{A,i}}$ and $f_{Y_i}$ are the probability density function of $Y_{A,i}$ in~\eqref{ith_compromised} and $Y_i$ in~\eqref{ith_noatt}, respectively.

\section{Information Theoretic Metrics}\label{sec_Information_Theoretic_Metrics}

The aim of the attacker is twofold: Firstly, it aims to compromise the data integrity of the state variables, and therefore, disrupting all processes that use the measurements such as the state estimator; and secondly, to guarantee a stealthy attack. In this work, instead of assuming a particular state estimation procedure, we consider information theoretic metrics for the disruption caused by the attacks and the detection. In doing so, we aim to establish the fundamental limits of the disruption and detection trade-off that a decentralized attack can achieve.

\subsection{Disruption Measure}
From~\eqref{20220518_2}, the attack vector $A^m$ in~\eqref{eq:obs_attack} is the result of all the attacks $A_i^m$, with $i\in \Kc$. Hence, the disruption caused by the attack $A_i^m \in \Ac_i$ in~\eqref{20220920_1} launched by attacher $i$ can be described by the mutual information between the state variables $X^n$ in~\eqref{Sigma_XX} and the vector of compromised measurements $Y_{A^m}$ in~\eqref{eq:obs_attack}, that is, $I(X^n;Y_{A^m})$ referred to as {\it global mutual information}. 
The rationale for measuring the disruption in terms of global mutual information stems from the case where the attacker considers the amount of information loss jointly caused by the vector of all the compromised measurements.
Note that global mutual information has previously been adopted to capture attack disruption in~\cite{SE_SGC_17,SE_TSG_19,YE_SGC_20}. The analytical expression of global mutual information $I(X^n;Y_{A^m})$ is~\cite[Prop. 2]{SE_TSG_19}
\begin{IEEEeqnarray}{rll}\label{eq_MI}
 I(X^n;Y^m_A) 
	= &\frac{1}{2}\textnormal{log}  \frac{\left| \Sigmam_{Y\!Y}  + \Sigmam_{A\!A}\right|}{\left| \sigma^2\textbf{I}_m + \Sigmam_{A\!A} \right|},
\end{IEEEeqnarray}
where the real $\sigma \in \mathds{R}_+$ is in~\eqref{EqZ}; the matrices $\Sigmam_{Y\!Y}$ and $\Sigmam_{A\!A}$ are in~\eqref{20220830_4} and~\eqref{eq:Gauss_attack}, respectively.

Alternatively, from~\eqref{distribution_A_i},~\eqref{ith_compromised} and~\eqref{ith_vector}, the disruption caused by the attack $A_i^m$ launched by attacher $i$, with $i \in \Kc$, can also be described by the mutual information between the state variables $X^n$ in~\eqref{Sigma_XX} and the compromised measurement $i$ denoted by $Y_{A_i}$ in~\eqref{ith_compromised}, that is, $I(X^n;Y_{A_i})$ referred to as {\it local mutual information}. The rationale of measuring the disruption in terms of local mutual information stems from the case where the attacker considers the amount of information loss locally caused by the measurement $i$ compromised in isolation.
The following proposition presents the analytical expression of the local mutual information.
\begin{prop}\label{prop_local_I}
	For all $i \in \{1,2,\ldots,m\}$, the random variables $X^n$ in~\eqref{Sigma_XX} and $Y_{A_i}$ in~\eqref{ith_compromised} satisfy
	\be\label{eq_local_I}
	I(X^n;Y_{A_i}) = \frac{1}{2}\log \left(1+\frac{ \ev_i^{\sf T}\Hm \Sigmam_{X\!X} \Hm^{\sf T}\ev_i }{\sigma^2 +v_i}\right),
	\ee
	where the matrix $\Hm$ is in~\eqref{eq:obs_noattack}; the real $\sigma \in \mathds{R}_+$ is in~\eqref{EqZ}; the matrix $\Sigmam_{X\!X}$ is in~\eqref{Sigma_XX}; and the real $v_i $ is in~\eqref{distribution_A_i}.
\end{prop}
\begin{proof}
	See Appendix A.
%	The random variables $X^n$ in~\eqref{Sigma_XX} and $Y_{A_i}$ in~\eqref{ith_compromised} form a multivariate Gaussian distribution such that $(X^n, Y_{A_i})\sim {\cal N}(\textbf{0}, \Sigmam)$
%	with 
%	\be
%	\Sigmam \eqdef 
%	\begin{bmatrix}
%		\Sigmam_{X\!X} & \Sigmam_{X\!X}	\hv_i^{\sf T}  \\
%		\hv_i \Sigmam_{X\!X}  &  \ev_i^{\sf T} \Sigmam_{Y\!Y} \ev_i+v_i\\
%	\end{bmatrix},
%	\ee
%	where $\hv_i$ is the $i$-th row of $\Hm$. Let $f_{X^n}, f_{Y_{A_i}}$ and $f_{X^n,Y_{A_i}}$ be the probability density functions of the random variables $X^n$, $Y_{A_i}$ and $(X^n,Y_{A_i})$, respectively. Hence, the following holds
%	\begin{IEEEeqnarray}{rll}
%		I(X^n; Y_{A_i})  \eqdef & \mathds{E}_{X^n, Y_{A_i}} \left[\textrm{log}  \dfrac{f_{X^n, Y_{A_i}}}{f_{X^n}f_{Y_{A_i}}}   \right]   \\
%		= &\dfrac{1}{2}  \log \dfrac{ \left|\Sigmam_{X\!X}\right| \left(\ev_i^{\sf T}\Sigmam_{Y\!Y}\ev_i +v_i\right) }{\left|\Sigmam\right|}   \\
%		= &\dfrac{1}{2}  \log \dfrac{ \ev_i^{\sf T}\Sigmam_{Y\!Y}\ev_i +v_i  }{\sigma^2+v_i}, \label{20220728_1}
%	\end{IEEEeqnarray}
%	where $\mathds{E}_{X^n, Y_{A_i}}$ denotes the expectation with respect to the random variables $(X^n, Y_{A_i})$; the equality in~\eqref{20220728_1} follows from $\left|\Sigmam\right| = \left|\Sigmam_{X\!X}\right| \left(\sigma^2+v_i\right)$. This completes the proof.
\end{proof}
\subsection{Detection Measure}
The Chernoff-Stein Lemma~\cite[Th. 11.7.3]{TM_ElementsofIT} states that the probability of detection of the random attack is characterized by the Kullback-Leibler (KL) divergence. The KL divergence is a measure of the difference between two distributions and in our setting it quantifies the deviation of the distribution under attacks with respect to the distribution under normal operation. The minimization of KL divergence is equivalent to the minimization of the asymptotic probability of attack detection~\cite{TM_ElementsofIT}. Note that KL divergence has previously been adopted to analyze attack detection in~\cite{SE_TSG_19} and~\cite{YE_SGC_20}.

Specifically, attacker $i$, with $i \in \Kc$, operates under the joint attack detection setting described in Section~\ref{sec_joint_detection}, that is, attacker $i$ controls the attack detection induced by the vector of compromised measurements $Y_A^m$. For the hypothesis test problem in~\eqref{lrt}, a small value
of the KL divergence between $P_{Y_A^m}$ in~\eqref{eq:obs_attack_distribution} and $P_{Y^m}$ in~\eqref{eq:obs_no_attack_distribution} denoted by $D(P_{Y_A^m}\|P_{Y^m})$ implies that the attack is unlikely to be detected by the LRT in~\eqref{lrt}~\cite[Th. 11.8.3]{TM_ElementsofIT}.
The analytical expression for the global KL divergence $D(P_{Y_A^m}||P_{Y^m})$ in our setting is~\cite[Prop.~1]{SE_TSG_19}
\begin{IEEEeqnarray}{rll}\label{eq_KL} 
	\!\!\!\!\!\!\!\!\! D( P_{Y_A^m}\|P_{Y^m}\!) 
	\!=\!   \frac{1}{2}\!\left( \!\log\frac{\left| \Sigmam_{Y\!Y} \right|}{\left|\Sigmam_{Y_A\!Y_A}\right|} +  \textnormal{tr} \left(\Sigmam_{Y\!Y}^{-1} \Sigmam_{A\!A} \right)\!\! \right)\!\!,
\end{IEEEeqnarray}
where the matrices $\Sigmam_{Y\!Y}$ and $\Sigmam_{A\!A}$ are in~\eqref{20220830_4} and \eqref{eq:Gauss_attack}, respectively.

Alternatively, the system operator can adopt the local attack detection described in Section~\ref{sec_local_detection}. In this case, attacker $i$, with $i\in\Kc$, only considers the attack detection based on measurement $i$.
For the hypothesis test problem in~\eqref{lrt_local}, a small value
of the KL divergence between $P_{Y_{A,i}}$ in~\eqref{ith_compromised} and $P_{Y_i}$ in~\eqref{ith_noatt}, denoted by $D(P_{Y_{A,i}}\|P_{Y_i})$ implies that the attack is unlikely to be detected by the LRT in~\eqref{lrt_local}~\cite[Th. 11.8.3]{TM_ElementsofIT}. 
The following proposition characterizes the KL divergence in this local detection setting.
\begin{prop}\label{prop_local_D}
	For all $i \in \{1,2,\ldots,m\}$, the random variables $Y_{A_i}$ in~\eqref{ith_compromised} and $Y_i$ in~\eqref{ith_noatt} satisfy
	\be\label{eq_local_D}
	\begin{split}
		D(P_{Y_{A, i}} \| P_{Y_i} )  
		= &\frac{1}{2}  \left(\frac{v_i}{ \ev_i^{\sf T}\Sigmam_{Y\!Y}  \ev_i  }  + \log \frac{\ev_i^{\sf T}\Sigmam_{Y\!Y}  \ev_i}{\ev_i^{\sf T}\Sigmam_{Y\!Y}  \ev_i+ v_i    } \right),
	\end{split}
	\ee
	where the matrix $\Sigmam_{Y\!Y}$ is defined in~\eqref{20220830_4}, and the $v_i \in \mathds{R}$ is defined in~\eqref{distribution_A_i}.
\end{prop}
\begin{proof}
	See Appendix B.
%	Let $f_{Y_{A_i}}$ and $f_{Y_i}$ be the probability density functions of the random variables $Y_{A_i}$ in~\eqref{ith_compromised} and $Y_i$ in~\eqref{ith_noatt}. Hence, the following holds
%	\begin{IEEEeqnarray}{rll}
%		D(P_{Y_{A, i}} \| P_{Y_i} ) 
%		\eqdef & \mathds{E}_{ Y_{A_i}} \left[\textrm{log} \dfrac{f_{Y_{A_i} }}{f_{Y_i}} \right]  \\
%		\nonumber
%		=& \frac{1}{2}  \left(\frac{v_i}{ \ev_i^{\sf T}\Sigmam_{Y\!Y}  \ev_i  }  + \log \frac{\ev_i^{\sf T}\Sigmam_{Y\!Y}  \ev_i}{\ev_i^{\sf T}\Sigmam_{Y\!Y}  \ev_i+ v_i    } \right).
%	\end{IEEEeqnarray}
%	This completes the proof.
\end{proof}

\section{Game Formulation}\label{sec_game_formulation}
The decentralized attacks described in Section~\ref{sec_Decentralized Data Injection Attacks} poses a framework for the exploration of game theoretic techniques~\cite{potentialgame}.
This section introduces three games with different aims of the attackers. For all $p \in \{1,2,3\}$, game~$p$, denoted by $\Gc_p$, is a game in normal form:
\be\label{gameformula_overall}
\gameNF_p = \left({\Kc}, \Vc , \{\phi^{(p)}_i\}_{i \in {\cal K}}\right),
\ee
where $\Kc$ is the set of players in~\eqref{set_K}, the set $\Vc \eqdef \left[0, +\infty \right)$ is the set of actions of player $i \in \Kc$, that is, the set of $v_i$ in~\eqref{distribution_A_i}, and $\phi^{(p)}_i: \mathds{R}^m \rightarrow \mathds{R}$ is a cost function such that for all $\vv \eqdef \left(v_1,v_2,\ldots,v_m\right) \in \Vc^m$, we define
	\be\label{game1_cost}
		\begin{aligned}
	\phi^{(1)}_i(\vv) \eqdef &I(X^n;Y_A^m) +\lambda D(P_{Y_A^m}\|P_{Y^m}) \\
	=&\frac{1}{2}(1-\lambda)\log \left|\Sigmam_{Y\!Y}  +  \sum_{i \in \Kc} v_i\ev_i \ev_i^{\sf T}\right|  \\
	&-  \frac{1}{2}\log \left|\sigma^2 \textbf{I}_m  + \sum_{i \in \Kc} v_i\ev_i \ev_i^{\sf T}\right| \\
	 &+ \frac{1}{2} \lambda \left( \log \left|\Sigmam_{Y\!Y}\right| + \textnormal{tr}\left(\Sigmam_{Y\!Y}^{-1}\sum_{i \in \Kc} v_i\ev_i \ev_i^{\sf T}\right)\right),\\
  \end{aligned}
  \ee
  	\be\label{game2_cost}
		\begin{aligned}
	\phi^{(2)}_i(\vv) \eqdef &I(X^n; Y_{A_i}) +\lambda D(P_{Y_A^m}\|P_{Y^m}) \\
	=& \frac{1}{2}\log \left(1+\frac{ \ev_i^{\sf T}\Hm \Sigmam_{X\!X} \Hm^{\sf T}\ev_i }{\sigma^2 +v_i}\right) \\
	&+ \frac{1}{2} \lambda\log \frac{\left| \Sigmam_{Y\!Y}  \right|}{\left| \Sigmam_{Y\!Y} + \sum_{j \in \Kc} v_j \ev_j \ev_j^{\sf T}\right|} \\
	&+ \frac{1}{2} \lambda  \textnormal{tr}\left(\Sigmam_{Y\!Y}^{-1}\sum_{j \in \Kc} v_j\ev_j \ev_j^{\sf T}\right),
   \end{aligned}
  \ee
    	\be\label{game3_cost}
		\begin{aligned}
	\phi^{(3)}_i(\vv) \eqdef &I(X^n; Y_{A}^m) +\lambda D(P_{Y_{A, i}} \| P_{Y_i} ) \\
	=&\frac{1}{2} \log \frac{\left|\Sigmam_{Y\!Y}  + \sum_{j \in \Kc} v_j\ev_j \ev_j^{\sf T}\right|}{\left|\sigma^2\textbf{I}_m +\sum_{j \in \Kc}v_j\ev_j \ev_j^{\sf T}\right|} \\
	&+ \frac{1}{2} \lambda \left( \frac{v_i}{\ev_i^{\sf T}\Sigmam_{Y\!Y} \ev_i} + \log\frac{\ev_i^{\sf T}\Sigmam_{Y\!Y} \ev_i}{\ev_i^{\sf T}\Sigmam_{Y\!Y} \ev_i+v_i} \right).
  \end{aligned}
  \ee
In the game $\Gc_1$, the cost function for attacker $i \in \Kc$ is $\phi^{(1)}_i$ in~\eqref{game1_cost}. In this game, the aim of the attackers is to jointly minimize the mutual information $I(X^n;Y_{A^m})$ in~\eqref{eq_MI} and KL divergence $D(P_{Y_A^m}\|P_{Y^m})$ in~\eqref{eq_KL}. 
Note that the weighting parameter $\lambda $ determines the tradeoff between the disruption caused by the attacks and the probability of attack detection. The rationale of the aim is to minimize the mutual information and the KL divergence globally.
%The rationale of considering global KL divergence is that apart from the disruption, attacker aims for the stealthiness of the attacks. From the Chernoff-Stein lemma in~\cite{TM_ElementsofIT}, minimizing the KL divergence is equivalent to minimizing the asymptotic attack detection probability. A weighted sum of global mutual information $I(X^n; Y_{A}^m)$ and global KL divergence $D(P_{Y_A^m}||P_{Y^m})$ yields the cost function for attacker $i$ in $\Gc_1$ in~\eqref{game1_cost}.
Similarly, in the game $\Gc_2$, the attacker $i$ aims at minimizing cost function $\phi^{(2)}_i$ in~\eqref{game2_cost} defined as a weighted sum of mutual information $I(X^n; Y_{A_i})$ in~\eqref{eq_local_I} and KL divergence $D(P_{Y_A^m}\|P_{Y^m})$ in~\eqref{eq_KL}. The rationale of the aim is to minimize the mutual information reduction caused by compromising local measurement but minimizing the KL divergence globally.  
In the game $\Gc_3$, the attacker $i$ aims to minimize cost function $\phi^{(3)}_i$ in~\eqref{game3_cost} defined as a weighted sum of mutual information $I(X^n; Y_A^m)$ in~\eqref{eq_MI} and KL divergence $D(P_{Y_{A, i}} \| P_{Y_i} )$ in~\eqref{eq_local_D}. The rationale of the aim is to minimize the mutual information globally but minimizing the KL divergence between locally compromised measurement and measurement without attack.
Note that~\eqref{game1_cost},~\eqref{game2_cost} and~\eqref{game3_cost} follow from plugging in the corresponding terms from~\eqref{eq_MI},~\eqref{eq_local_I},~\eqref{eq_KL} and~\eqref{eq_local_D} into~\eqref{game1_cost},~\eqref{game2_cost} and~\eqref{game3_cost}, respectively, and replacing $\Sigmam_{A\!A}$ with $\sum_{i\in \Kc}  v_i \ev_i\ev_i^{\sf T}$ as shown in~\eqref{20220519_6} and~\eqref{set_K}.
%In the game $\Gc_2$, instead of considering mutual information globally, the first term of the cost function $\phi^{(2)}_i$ in~\eqref{game2_cost} for attacker $i$, with $i \in \Kc$, represents the local mutual information between the state variables $X^n$ and the compromised measurement $i$ denoted by $Y_{A_i}$ in~\eqref{ith_compromised}, that is, $I(X^n; Y_{A_i})$ in~\eqref{eq_local_I}. The rationale of measuring the disruption in terms of local mutual information stems from the attacker considers the amount of information loss caused by the measurement compromised by itself. The second term of the cost function $\phi^{(2)}_i$ represents the global KL divergence $D(P_{Y_A^m}||P_{Y^m})$.

%In the game $\Gc_3$, the first term of the cost function~$\phi^{(3)}_i$ in~\eqref{game3_cost} for attacker $i$, with $i \in \Kc$, represents the global mutual information $I(X^n;Y_{A^m})$. The second term represents the local KL divergence between the distribution of $Y_{A_i}$ in~\eqref{ith_compromised} and the measurement $i$ without attacks~$Y_i$ in~\eqref{ith_noatt}, that is, $D(P_{Y_{A, i}} \| P_{Y_i} )$ in~\eqref{eq_local_D}. That is, the attacker considers the stealthiness of the attack launched by itself.  
 	
\subsection{Convexity}
The following proposition characterizes the convexity of the cost functions $\phi_i^{(p)}$ in~\eqref{game1_cost}~\eqref{game2_cost}~\eqref{game3_cost}, with~$i ~\in~\Kc$ and $p~\in\{1,2,3\}$.
 \begin{prop}\label{prop_game1_cov}
	Let~$\lambda \geq 1$. For all~$i \in~\Kc$ and for all~$p~\in~\{1,2,3\}$, the cost functions~$\phi^{(p)}_i$ are convex in $v_i$.
\end{prop}
\begin{proof}
	The proof is divided into three parts. For all~$p \in \{1,2,3\}$, part $p$ presents the proof of convexity of $\phi^{(p)}_i$. 
	
	The first part is as follows. From~\cite[Sec. 3.1.5]{boyd_cvx} and the fact that $\Sigmam_{Y\!Y} + \sum_{i \in \Kc} v_i\ev_i \ev_i^{\sf T} \in~\Sm_{++}^m$ and $\sigma^2 \textbf{I}_m  + \sum_{i \in \Kc} v_i\ev_i \ev_i^{\sf T}   \in ~\Sm_{++}^m$, the terms $\textrm{log} \left| \Sigmam_{YY} +\sum_{i \in \Kc} v_i\ev_i \ev_i^{\sf T} \right|$ and $\log \left|\sigma^2 \textbf{I}_m  + \sum_{i \in \Kc} v_i\ev_i \ev_i^{\sf T}\right|$ are concave. Therefore, the terms 
	\be
	\left(1-\lambda\right)\textrm{log} \left| \Sigmam_{YY} +\sum_{i \in \Kc} v_i\ev_i \ev_i^{\sf T} \right|,
	\ee
	 and $-\log \left|\sigma^2 \textbf{I}_m  + \sum_{i \in \Kc} v_i\ev_i \ev_i^{\sf T}\right|$ in~\eqref{game1_cost} are convex. Given that the trace is a linear operator, the term $\textnormal{tr}\left(\Sigmam_{Y\!Y}^{-1}\sum_{i \in \Kc} v_i\ev_i \ev_i^{\sf T}\right)  $ is linear with respect to $v_i$. From~\cite[Sec. 3.2.1]{boyd_cvx}, the non-negative weighted sums of convex functions is convex, that is, the cost function $\phi_i^{(1)} $ is convex. 
	
The second part is as follows. Note that the term $\log \left(1+\frac{ \ev_i^{\sf T}\Hm \Sigmam_{X\!X} \Hm^{\sf T}\ev_i }{\sigma^2 +v_i}\right)$ is convex in $v_i$. From~\cite[Sec.~3.1.5]{boyd_cvx} and the fact that $\Sigmam_{Y\!Y} + \sum_{j \in \Kc} v_j\ev_j \ev_j^{\sf T} \in \Sm_{++}^m$, the term $\textrm{log} \left| \Sigmam_{YY} +\sum_{j \in \Kc} v_j\ev_j \ev_j^{\sf T} \right|$ is concave. Hence, the term $-\log  \left| \Sigmam_{Y\!Y} + \sum_{j \in \Kc} v_j \ev_j \ev_j^{\sf T}\right| $ is convex. Given that trace is a linear operator, the term $\textnormal{tr}\left(\Sigmam_{Y\!Y}^{-1}\sum_{j \in \Kc} v_j\ev_j \ev_j^{\sf T}\right)  $ is linear with respect to $v_i$. From~\cite[Sec. 3.2.1]{boyd_cvx}, the non-negative weighted sums of convex functions is convex, that is, the cost function $\phi_i^{(2)} $ in~\eqref{game2_cost} is convex in $v_i$.

The third part is as follows. Note that the equality in~\eqref{game3_cost} can be rewritten as follows:

\begin{IEEEeqnarray}{rll}
	\nonumber
	\!\!\!\!  \phi^{(3)}_i(\vv) 
	=&\frac{1}{2} \log \frac{\left|\Hm\Sigmam_{XX}\Hm^{\sf T}+ \sigma^2\textbf{I}_m+ \sum_{j \in \Kc} v_j\ev_j \ev_j^{\sf T}\right|}{\left|\sigma^2\textbf{I}_m +\sum_{j \in \Kc}v_j\ev_j \ev_j^{\sf T}\right|}\label{20220901_1} \\
	&+ \frac{1}{2} \lambda \left( \frac{v_i}{\ev_i^{\sf T}\Sigmam_{Y\!Y} \ev_i} + \log\frac{\ev_i^{\sf T}\Sigmam_{Y\!Y} \ev_i}{\ev_i^{\sf T}\Sigmam_{Y\!Y} \ev_i+v_i} \right)\\
	\nonumber
	=&\frac{1}{2} \log \left|\Hm\Sigmam_{X\!X}\Hm^{\sf}\left(\sigma^2\textbf{I}_m + \sum_{j \in \Kc} v_j\ev_j \ev_j^{\sf T}\right)^{-1}\!\!\!\!+\!\!\textbf{I}_m\right| \label{20220728_10} \\
		 &+ \frac{1}{2} \lambda \left( \frac{v_i}{\ev_i^{\sf T}\Sigmam_{Y\!Y} \ev_i} + \log\frac{\ev_i^{\sf T}\Sigmam_{Y\!Y} \ev_i}{\ev_i^{\sf T}\Sigmam_{Y\!Y} \ev_i+v_i} \right)\\
		 \nonumber
		 =& \frac{1}{2} \log  \left| \frac{1}{\sigma^2 + v_i} \Hm\Sigmam_{X\!X}\Hm^{\sf T} \ev_i \ev_i^{\sf T} \right. \label{20220830_6} \\
		 \nonumber
		& \ \ \ \  \left.+ \Hm\Sigmam_{X\!X}\Hm^{\sf T} \sum_{j \in \Kc \setminus \{i\}}^m \frac{1}{\sigma^2 + v_j} \ev_j \ev_j^{\sf T}    + \textnormal{\textbf{I}}_m\right|  \\ 
		& +  \frac{1}{2} \lambda \left( \frac{v_i}{\ev_i^{\sf T}\Sigmam_{Y\!Y} \ev_i} + \log\frac{\ev_i^{\sf T}\Sigmam_{Y\!Y} \ev_i}{\ev_i^{\sf T}\Sigmam_{Y\!Y} \ev_i+v_i} \right),
\end{IEEEeqnarray} 
where the equality in~\eqref{20220901_1} follows from plugging~\eqref{20220830_4} into the first term of~\eqref{game3_cost}; the equality in~\eqref{20220728_10} follows from canceling the term $\left| \sigma^2\textbf{I}_m + \sum_{j \in \Kc} v_j\ev_j \ev_j^{\sf T}\right|$ from the numerator and denominator of the first term in~\eqref{20220901_1}.
	Let the matrix $\Am \eqdef \Hm\Sigmam_{X\!X}\Hm^{\sf T} \sum_{j \in {\cal K} \setminus \{i\}} \dfrac{1}{\sigma^2 + v_j} \ev_j \ev_j^{\sf T}    + \textbf{I}_m$ and $a_i~\eqdef~\ev_i^{\sf T}\Am^{-1}\Hm\Sigmam_{X\!X}\Hm^{\sf T}\ev_i $. For all $i \in \Kc$, the following holds for the first term in~\eqref{20220830_6}
\begin{IEEEeqnarray}{rll}
			 &\frac{\partial }{\partial v_i} \textnormal{log} \left|\dfrac{1}{\sigma^2 + v_i} \Hm\Sigmam_{X\!X}\Hm^{\sf T}\ev_i \ev_i^{\sf T} + \Am\right|\label{20220613_1}\\
			\nonumber
			 =& - \frac{a_i}{(\sigma^2+v_i)\left(\sigma^2+v_i+a_i\right)},  \ \ \mbox{and}\\
			&\frac{\partial^2 }{\partial v_i^2} \textnormal{log} \left|\frac{1}{\sigma^2 + v_i} \Hm\Sigmam_{X\!X}\Hm^{\sf T}\ev_i \ev_i^{\sf T} + \Am\right| \label{20220613_2}\\
			\nonumber =&\frac{a_i\left(2(\sigma^2+v_i)+a_i\right)}{(\sigma^2+v_i)^2\left(\sigma^2+v_i+a_i\right)^2},
\end{IEEEeqnarray} 
where the equalities in~\eqref{20220613_1} and~\eqref{20220613_2} follow from Appendix C. Note that $a_i > 0$ and $v_i \geq 0$. Hence, the first derivative in~\eqref{20220613_1} is negative and the second derivative in~\eqref{20220613_2} is positive. Therefore, the term $\textnormal{log} \left| \frac{1}{\sigma^2 + v_i} \Hm\Sigmam_{X\!X}\Hm^{\sf T} \ev_i \ev_i^{\sf T}  \right.\\
\left.+ \Hm\Sigmam_{X\!X}\Hm^{\sf T} \sum_{j \in \Kc \setminus \{i\}}^m \frac{1}{\sigma^2 + v_j} \ev_j \ev_j^{\sf T}    + \textnormal{\textbf{I}}_m\right|$ in~\eqref{20220830_6} is convex in $v_i$. The term $\frac{v_i}{\ev_i^{\sf T}\Sigmam_{Y\!Y} \ev_i}$ is linear with respect to $v_i$ and the term $\log\frac{\ev_i^{\sf T}\Sigmam_{Y\!Y} \ev_i}{\ev_i^{\sf T}\Sigmam_{Y\!Y} \ev_i+v_i}$ is convex in $v_i$. From~\cite[Sec. 3.2.1]{boyd_cvx}, the non-negative weighted sums of convex functions is convex, that is, the cost function $\phi_i^{(3)} $ in~\eqref{game3_cost} is convex in $v_i$.
This completes the proof.
\end{proof}
 
\subsection{Best Responses}\label{sec_bestresponse}
From the game formulation in~\eqref{gameformula_overall}, the set of actions $\Vc$ of player $i \in \Kc$ is the set of admissible variances of random variables $A_i$~in~\eqref{distribution_A_i} with $i \in \Kc$.
 The underlying assumption in the following is that, for all $i \in \Kc$, given an action profile $\vv_{-i}~\eqdef~\left(v_1,v_2,\ldots,v_{i-1},v_{i+1},\ldots,v_{m-1},v_m\right) \in \Vc^{m-1}$ formed by players $j$, with $j\in \Kc \setminus \{i\}$, player $i$ adopts an action $v_i$ such that the cost function $\phi^{(p)}_i \left(\vv \right)$ is minimized, that is,
 \be\label{BR_v_i_overall}
 v_i  \in \textrm{BR}_i^{(p)}\left(\vv_{i-1}\right),
 \ee  
 where the correspondence $\textrm{BR}_i^p$: $\Vc^{m-1}   \rightarrow \mathscr{B}\left(\left[0,+\infty\right]\right)$, where $\mathscr{B}$ is the borel set, is the best response correspondence, that is,
 \begin{IEEEeqnarray}{rll}
 	\textrm{BR}_i^{(p)}\left(\vv_{i-1}\right) = \argmin_{v_i \in  \Vc }  \phi^{(p)}_i \left(\vv\right).
 \end{IEEEeqnarray}
 
The following theorems provide the best responses of player $i$ in $\gameNF_p$, with $i \in \Kc$ and $p~\in~\{1,2,3\}$.

\begin{theorem}\label{Theorem_BR_game1}
		For all $i \in \Kc$, the set of best responses of player $i$ to the complementary action profile $\vv_{-i} \in~\Vc^{m-1}$ in~$\gameNF_1$ in~\eqref{gameformula_overall} is
		\begin{IEEEeqnarray}{rll}\label{eq_BR_game1}
			& \textrm{\textnormal{BR}}_i^{(p) } \left(\vv_{i-1}\right) =  \Bigg\{\!l\!\in\![0,+\infty)\!: \\
            \nonumber
			&l\!=\!\frac{1}{2}\!\sqrt{\! \Bigl(\! \frac{\beta_i + \alpha_i \sigma^2  \beta_i - \alpha_i}{\beta_i\alpha_i}\Bigl)^2 \!\! -  \frac{4 \left( \!\beta_i \sigma^2 \!- \! \alpha_i\sigma^2  \! +\!\frac{1}{\lambda}(\alpha_i \sigma^2 -1) \right)}{\beta_i\alpha_i}\! } \\
            \nonumber
            & \ \ \ -  \dfrac{ \beta_i\!+\!\alpha_i \sigma^2  \beta_i\!-\!\alpha_i }{2\beta_i\alpha_i}, \\
            \nonumber
			& \mbox{with} \ \ \alpha_i \eqdef  \ev_i^{\sf T} \Big( \Sigmam_{YY}  + \sum_{j \in {\cal K} \setminus \{i\}}  v_{j} \ev_j \ev_j^{\sf T}  \Big)^{-1} \ev_i, \ \ \mbox{and} \\
            \nonumber
            & \beta_i \eqdef  \ev_i^{\sf T} \Sigmam_{YY}^{-1}\ev_i  \label{20220831_1}\Bigg\}.
		\end{IEEEeqnarray} 
\end{theorem}
\begin{proof}
	From Proposition~\ref{prop_game1_cov}, the cost function for player $i \in \Kc$ in $\Gc_1$ is convex. Therefore, the best responses are achieved on the first saddle point. This completes the proof.
\end{proof}

\begin{theorem}\label{Theorem_BR_game2}
	For all $i \in \Kc$, the set of best responses of player $i$ to the complementary action profile $\vv_{-i} \in~\Vc^{m-1}$ in~$\gameNF_2$ in~\eqref{gameformula_overall} is
\begin{IEEEeqnarray}{rll}\label{eq_BR_game2} 
	 \textrm{\textnormal{BR}}_i^{(p) }\left(\vv_{i-1}\right)  & = \Bigg\{l\in[0,+\infty): \\
     \nonumber
    & \frac{-\ev_i^{\sf T}\Hm \Sigmam_{X\!X} \Hm^{\sf T}\ev_i}{\left(\sigma^2+l\right)\left(  \ev_i^{\sf T} \Sigmam_{Y\!Y}  \ev_i + l\right)} - \lambda   \frac{\alpha_i}{1+l\alpha_i} + \lambda \beta_i = 0, \\
    \nonumber
    & \mbox{with} \ \ \alpha_i  \ \mbox{and} \ \beta_i \ in~\eqref{20220831_1}\Bigg\}.
\end{IEEEeqnarray}
\end{theorem}
 \begin{proof}
 	From Proposition~\ref{prop_game1_cov}, the cost function for player $i \in \Kc$ in $\Gc_2$ is convex. Therefore, the best responses are achieved on the first saddle point. This completes the proof.
 \end{proof}

\begin{theorem}\label{Theorem_BR_game3}
		For all $i \in \Kc$, the set of best responses of player $i$ to the complementary action profile $\vv_{-i} \in~\Vc^{m-1}$ in~$\gameNF_3$ in~\eqref{gameformula_overall} is
		\begin{IEEEeqnarray}{rll}\label{eq_BR_game3}
		& \textrm{\textnormal{BR}}_i^{(p) }\!\!\left(\vv_{i-1}\right)  = \Bigg\{\!l\!\in\![0,+\infty)\!: \\
        \nonumber
		& \frac{ \lambda l }{\left(\ev_i^{\sf T} \Sigmam_{Y\!Y}  \ev_i+l \right)\ev_i^{\sf T} \Sigmam_{Y\!Y}  \ev_i}- \frac{ \gamma_i}{(\sigma^2+l )\left(\sigma^2+l +\alpha_i\right)} = 0, \\
		\nonumber
		& \mbox{with} \ \ \alpha_i \ \ in~\eqref{20220831_1}, \ \ \mbox{and} \ \ \gamma_i \ \ \textnormal{as} \ \ \\
        \nonumber
		& \ev_i^{\sf T}   \left( \! \Hm\Sigmam_{X\!X}\Hm^{\sf T}\!\!\!\!\!\!\!\!\!\!\sum_{j \in {\cal K} \setminus \{i\}}\!\!\!\dfrac{1}{\sigma^2 + v_j} \ev_j \ev_j^{\sf T}\! \!+\! \!\textnormal{\textbf{I}}_m\!\!\right)^{\!\!\!-1}\!\!\!\!\!\!\!\!\!\!\!\!\Hm \Sigmam_{X\!X} \Hm^{\sf T}\ev_i\Bigg\}.
	\end{IEEEeqnarray}
\end{theorem}
\begin{proof}
		From Proposition~\ref{prop_game1_cov}, the cost function for player $i \in \Kc$ in $\Gc_3$ is convex. Therefore, the best responses are achieved on the first saddle point. This completes the proof.
\end{proof}

\twocolumn
\subsection{Potential Games}\label{sec_potentialgames}
The following propositions provide the potential functions of the games that arise in this setting. The property of being potential games allows the attackers in decentralized attacks to construct attacks with game theoretic techniques in this security setting.
\begin{prop}\label{prop_game1_potential}
	The game $\gameNF_1$ is a potential game with a potential function $\psi_1$: $\Vc^m \rightarrow \mathds{R}_+$, such that for all $\vv \in \Vc^m$, 
		\begin{IEEEeqnarray}{rll}\label{potential_game1}
		\psi_1(\vv) \eqdef &\frac{1}{2}(1-\lambda)\log \left|\Sigmam_{Y\!Y}  +  \sum_{i \in \Kc} v_i\ev_i \ev_i^{\sf T}\right| \\
		\nonumber
		&-  \frac{1}{2}\log \left|\sigma^2 \textbf{I}_m  + \sum_{i \in \Kc} v_i\ev_i \ev_i^{\sf T}\right| \\
		\nonumber
		&+ \frac{1}{2} \lambda \left( \log \left|\Sigmam_{Y\!Y}\right| + \textnormal{tr}\left(\Sigmam_{Y\!Y}^{-1}\sum_{i \in \Kc} v_i\ev_i \ev_i^{\sf T}\right)\right)\!\!.
	\end{IEEEeqnarray}
\end{prop}
\begin{proof}
The proof follows from the observation that all the players have the same cost function $\phi_i^{(1)}$ in~\eqref{game1_cost}, with $i \in \Kc$~\cite{potentialgame}. This completes the proof.
\end{proof}

\begin{prop}\label{prop_game2_potential}
	The game $\gameNF_2$ is a potential game with a potential function $\psi_2$: $\Vc^m \rightarrow \mathds{R}_+$, such that for all $\vv \in \Vc^m$, 
	\begin{IEEEeqnarray}{rll}\label{potential_game2}
		\psi_2(\vv) \eqdef &\frac{1}{2}\sum_{j \in \Kc} \log \left(1+\frac{ \ev_j^{\sf T}\Hm \Sigmam_{X\!X} \Hm^{\sf T}\ev_j }{\sigma^2 +v_j}\right)  \\
		\nonumber
		&+ \frac{1}{2} \lambda\log \frac{\left| \Sigmam_{Y\!Y}  \right|}{\left| \Sigmam_{Y\!Y} + \sum_{j \in \Kc} v_j \ev_j \ev_j^{\sf T}\right|} \\
		\nonumber
		&+ \frac{1}{2} \lambda  \textnormal{tr}\left(\Sigmam_{Y\!Y}^{-1}\sum_{j \in \Kc} v_j\ev_j \ev_j^{\sf T}\right).
	\end{IEEEeqnarray}
\end{prop}
\begin{proof}
	We follow the same method as in~\cite{potentialgame}, the proof follows from evaluating the decrease of the cost function $\phi_i^{(2)}$ in~\eqref{game2_cost} and the potential function $\psi_2$ in~\eqref{potential_game2} when deviating the action of player $i\in \Kc$.
	
Note that
\begin{IEEEeqnarray}{rll} 
&\phi_i^{(2)} (\vv ) - \phi_i^{(2)} (\vv_{-i}, x )  \\
%= &  \log \left(1+\frac{ \ev_i^{\sf T}\Hm \Sigmam_{X\!X} \Hm^{\sf T}\ev_i }{\sigma^2 +v_i}\right) \\
%    \nonumber
%    &+   \lambda\log \frac{\left| \Sigmam_{Y\!Y}  \right|}{\left| \Sigmam_{Y\!Y} + \sum_{j \in \Kc} v_j \ev_j \ev_j^{\sf T}\right|}  +  \lambda  \textnormal{tr}\left(\Sigmam_{Y\!Y}^{-1}\sum_{j \in \Kc} v_j\ev_j \ev_j^{\sf T}\right) \\
%    \nonumber
%    & - \log \left(1+\frac{ \ev_i^{\sf T}\Hm \Sigmam_{X\!X} \Hm^{\sf T}\ev_i }{\sigma^2 +x}\right) \\
%    \nonumber
%    &-   \lambda\log \frac{\left| \Sigmam_{Y\!Y}  \right|}{\left| \Sigmam_{Y\!Y} + x\ev_i\ev_i^{\sf T}+\sum_{j \in \Kc \setminus \{i\}} v_j \ev_j \ev_j^{\sf T}\right|}  \\
%    \nonumber
%    &-  \lambda  \textnormal{tr}\left(\Sigmam_{Y\!Y}^{-1}\left(x\ev_i\ev_i^{\sf T}+\sum_{j \in \Kc \setminus \{i\}} v_j \ev_j \ev_j^{\sf T}\right)\right)\\  
    \nonumber
    =&  \log \left(1+\frac{ \ev_i^{\sf T}\Hm \Sigmam_{X\!X} \Hm^{\sf T}\ev_i }{\sigma^2 +v_i}\right) - \log \left(1+\frac{ \ev_i^{\sf T}\Hm \Sigmam_{X\!X} \Hm^{\sf T}\ev_i }{\sigma^2 +x}\right)\\
    \nonumber
    &+   \lambda\log \frac{\left| \Sigmam_{Y\!Y} + x\ev_i\ev_i^{\sf T}+\sum_{j \in \Kc \setminus \{i\}} v_j \ev_j \ev_j^{\sf T}\right|}{\left| \Sigmam_{Y\!Y} + \sum_{j \in \Kc} v_j \ev_j \ev_j^{\sf T}\right|}   \\
    &+  \lambda  (v_i - x)\textnormal{tr}\left(\Sigmam_{Y\!Y}^{-1} \ev_i \ev_i^{\sf T}\right)\label{20220927_1}, 
\end{IEEEeqnarray} 
and that
\begin{IEEEeqnarray}{rll}
    	&	\psi_i^{(2)} (\vv ) - \psi_i^{(2)} (\vv_{-i}, x ) \\
% 	=&\frac{1}{2}\log \left(1+\frac{ \ev_i^{\sf T}\Hm \Sigmam_{X\!X} \Hm^{\sf T}\ev_i }{\sigma^2 +v_i}\right) \\ 
% 	  \nonumber
%  &+\frac{1}{2} \sum_{j \in \Kc  \setminus \{i\}} \log \left(1+\frac{ \ev_j^{\sf T}\Hm \Sigmam_{X\!X} \Hm^{\sf T}\ev_j }{\sigma^2 +v_j}\right)  \\
% 	\nonumber
% 	&+ \frac{1}{2} \lambda\log \frac{\left| \Sigmam_{Y\!Y}  \right|}{\left| \Sigmam_{Y\!Y}   + \sum_{j \in \Kc  } v_j \ev_j \ev_j^{\sf T}\right|} \\
% 	\nonumber
% 	&+ \frac{1}{2} \lambda  \textnormal{tr}\left(\Sigmam_{Y\!Y}^{-1}\left( v_i \ev_i \ev_i^{\sf T} + \sum_{j \in \Kc \setminus \{i\}} v_j \ev_j \ev_j^{\sf T}\right)\right)\\
% 	  \nonumber
%  &-\frac{1}{2}\log \left(1+\frac{ \ev_i^{\sf T}\Hm \Sigmam_{X\!X} \Hm^{\sf T}\ev_i }{\sigma^2 +x}\right) \\
%    \nonumber
%     &-\frac{1}{2} \sum_{j \in \Kc  \setminus \{i\}} \log \left(1+\frac{ \ev_j^{\sf T}\Hm \Sigmam_{X\!X} \Hm^{\sf T}\ev_j }{\sigma^2 +v_j}\right)  \\
%       \nonumber
%   	&+ \frac{1}{2} \lambda\log \frac{\left| \Sigmam_{Y\!Y}  \right|}{\left| \Sigmam_{Y\!Y} + x \ev_i \ev_i^{\sf T} + \sum_{j \in \Kc \setminus \{i\}} v_j \ev_j \ev_j^{\sf T}\right|} \\
%  \nonumber
%  &+ \frac{1}{2} \lambda  \textnormal{tr}\left(\Sigmam_{Y\!Y}^{-1}\left( x \ev_i \ev_i^{\sf T} + \sum_{j \in \Kc \setminus \{i\}} v_j \ev_j \ev_j^{\sf T}\right)\right)\\
     \nonumber
 =&  \log \left(1+\frac{ \ev_i^{\sf T}\Hm \Sigmam_{X\!X} \Hm^{\sf T}\ev_i }{\sigma^2 +v_i}\right) - \log \left(1+\frac{ \ev_i^{\sf T}\Hm \Sigmam_{X\!X} \Hm^{\sf T}\ev_i }{\sigma^2 +x}\right)\\
 \nonumber
 &+   \lambda\log \frac{\left| \Sigmam_{Y\!Y} + x\ev_i\ev_i^{\sf T}+\sum_{j \in \Kc \setminus \{i\}} v_j \ev_j \ev_j^{\sf T}\right|}{\left| \Sigmam_{Y\!Y} + \sum_{j \in \Kc} v_j \ev_j \ev_j^{\sf T}\right|}   \\
 &+  \lambda  (v_i - x)\textnormal{tr}\left(\Sigmam_{Y\!Y}^{-1} \ev_i \ev_i^{\sf T}\right)\label{20220927_2}.
\end{IEEEeqnarray}  
For every player $i \in \Kc$ and for every $\vv_{-i} \in \Vc^{m-1} $, let us assume $\phi_i^{(2)} (\vv ) - \phi_i^{(2)} (\vv_{-i}, x ) \leq 0$. Hence, 
%\begin{IEEEeqnarray}{rll}
%	\nonumber
% &  \log \left(1+\frac{ \ev_i^{\sf T}\Hm \Sigmam_{X\!X} \Hm^{\sf T}\ev_i }{\sigma^2 +v_i}\right) - \log \left(1+\frac{ \ev_i^{\sf T}\Hm \Sigmam_{X\!X} \Hm^{\sf T}\ev_i }{\sigma^2 +x}\right)\\
%\nonumber
%&+   \lambda\log \frac{\left| \Sigmam_{Y\!Y} + x\ev_i\ev_i^{\sf T}+\sum_{j \in \Kc \setminus \{i\}} v_j \ev_j \ev_j^{\sf T}\right|}{\left| \Sigmam_{Y\!Y} + \sum_{j \in \Kc} v_j \ev_j \ev_j^{\sf T}\right|}   \\
%&+  \lambda  (v_i - x)\textnormal{tr}\left(\Sigmam_{Y\!Y}^{-1} \ev_i \ev_i^{\sf T}\right) \leq 0.
%\end{IEEEeqnarray} 
from the equalities in~\eqref{20220927_1} and~\eqref{20220927_2}, it follows that
\be
\psi_i^{(2)} (\vv ) - \psi_i^{(2)} (\vv_{-i}, x ) \leq 0.
\ee
For every player $i \in \Kc$ and for every $\vv_{-i} \in \Vc^{m-1} $, let us assume $\psi_i^{(2)} (\vv ) - \psi_i^{(2)} (\vv_{-i}, x ) \leq 0$. Hence, from the equalities in~\eqref{20220927_1} and~\eqref{20220927_2}, it follows that
\be
\phi_i^{(2)} (\vv ) - \phi_i^{(2)} (\vv_{-i}, x ) \leq 0.
\ee   
Therefore, it holds that $\forall v_i \in \Vc \ \mbox{and} \ x \in \Vc$,
\begin{IEEEeqnarray}{rll}\label{ineqn1_game2}
 \phi_i^{(2)} (\vv ) - \phi_i^{(2)} (\vv_{-i}, x ) \leq 0 
\end{IEEEeqnarray} 
if and only if
\begin{IEEEeqnarray}{rll}
 \psi_2(\vv)\label{ineqn2_game2} - \psi_2(\vv_{-i}, x) \leq 0.
\end{IEEEeqnarray} 
This completes the proof.
\end{proof}

\begin{prop}\label{prop_game3_potential}
	The game $\gameNF_3$ is a potential game with a potential function $\psi_3$: $\Vc^m \rightarrow \mathds{R}_+$, such that for all $\vv \in \Vc^m$, 
	\begin{IEEEeqnarray}{rll}\label{potential_game3}
	\psi_3(\vv) \eqdef &\frac{1}{2} \log \frac{\left|\Sigmam_{Y\!Y}  + \sum_{j \in \Kc} v_j\ev_j \ev_j^{\sf T}\right|}{\left|\sigma^2\textbf{I}_m +\sum_{j \in \Kc}v_j\ev_j \ev_j^{\sf T}\right|} \\
	\nonumber
	&+ \frac{1}{2} \lambda \sum_{j \in \Kc}\left( \frac{v_j}{\ev_j^{\sf T}\Sigmam_{Y\!Y} \ev_j} + \log\frac{\ev_j^{\sf T}\Sigmam_{Y\!Y} \ev_j}{\ev_j^{\sf T}\Sigmam_{Y\!Y} \ev_j+v_j} \right).
\end{IEEEeqnarray}
\end{prop}
\begin{proof}
	We follow the same method as in~\cite{potentialgame}, the proof follows from evaluating the decrease of the cost function $\phi_i^{(3)}$ in~\eqref{game3_cost} and the potential function $\psi_3$ in~\eqref{potential_game3} when deviating the action of player $i\in \Kc$.
	
	Note that
	\begin{IEEEeqnarray}{rll} 
		  &\phi_i^{(3)} (\vv ) - \phi_i^{(3)} (\vv_{-i}, x )  \\
	=& \frac{1}{2} \log \frac{ 1 +v_i\textnormal{tr} \left(\Sigmam_{YY}^{-1} \ev_i \ev_i^{\sf T} \right)}{1 +x\textnormal{tr} \left(\Sigmam_{YY}^{-1} \ev_i \ev_i^{\sf T} \right)} - \frac{1}{2} \log \frac{  \sigma^2  + v_i }{ \sigma^2 + x }  \\
	\nonumber
	& + \frac{1}{2} \lambda \left( \frac{v_i -x}{\ev_i^{\sf T}\Sigmam_{Y\!Y} \ev_i} + \log\frac{\ev_i^{\sf T}\Sigmam_{Y\!Y} \ev_i+x}{\ev_i^{\sf T}\Sigmam_{Y\!Y} \ev_i+v_i} \right)\label{20220927_3},
	\end{IEEEeqnarray}	
and that	
\begin{IEEEeqnarray}{rll} 
&\psi_i^{(3)} (\vv ) - \psi_i^{(3)} (\vv_{-i}, x ) \\
=&\frac{1}{2} \log \frac{  1 + v_i \textnormal{tr} \left(\Sigmam_{Y\!Y}^{-1} \ev_i \ev_i^{\sf T}\right)}{1 + x \textnormal{tr} \left(\Sigmam_{Y\!Y}^{-1} \ev_i \ev_i^{\sf T}\right)}  - \frac{1}{2} \log \frac{ \sigma^2  + v_i  }{ \sigma^2 + x  }\\
\nonumber
&+ \frac{1}{2} \lambda \left( \frac{v_i -x}{\ev_i^{\sf T}\Sigmam_{Y\!Y} \ev_i} + \log\frac{\ev_i^{\sf T}\Sigmam_{Y\!Y} \ev_i + x}{\ev_i^{\sf T}\Sigmam_{Y\!Y} \ev_i+v_i} \right). \label{20220927_4}
\end{IEEEeqnarray}		
For every player $i$ and for every $\vv_{-i} \in \Vc^{m-1} $, let us assume $\phi_i^{(3)} (\vv ) - \phi_i^{(3)} (\vv_{-i}, x ) \leq 0$. Hence, from the equalities in~\eqref{20220927_3} and~\eqref{20220927_4}, it follows that
\begin{IEEEeqnarray}{rll} 
	&\psi_i^{(3)} (\vv ) - \psi_i^{(3)} (\vv_{-i}, x )  \leq 0.
\end{IEEEeqnarray}
For every player $i$ and for every $\vv_{-i} \in \Vc^{m-1} $, let us assume $\psi_i^{(3)} (\vv ) - \psi_i^{(3)} (\vv_{-i}, x ) \leq 0$. Hence, from the equalities in~\eqref{20220927_3} and~\eqref{20220927_4}, it follows that
\be
\phi_i^{(3)} (\vv ) - \phi_i^{(3)} (\vv_{-i}, x ) \leq 0.
\ee   
Therefore, it holds that $\forall v_i \in \Vc \ \mbox{and} \ x \in \Vc$,
\begin{IEEEeqnarray}{rll} 
	\phi_i^{(3)} (\vv ) - \phi_i^{(3)} (\vv_{-i}, x ) \leq 0 
\end{IEEEeqnarray} 
if and only if
\begin{IEEEeqnarray}{rll}
	\psi_3(\vv)  - \psi_3(\vv_{-i}, x) \leq 0.
\end{IEEEeqnarray} 
This completes the proof.
\end{proof}

\section{Existence and Achievability of Nash Equilibrium}\label{sec_NE}
\subsection{Existence of an Nash Equilibrium}
From the propositions~\ref{prop_game1_potential},~\ref{prop_game2_potential} and~\ref{prop_game3_potential} in Section~\ref{sec_potentialgames}, the games in~\eqref{gameformula_overall} are potential games where the interaction between the players yields the best responses described in Section~\ref{sec_bestresponse}. A game solution that is particularly relevant for this analysis is the Nash Equilibrium (NE). 
\begin{definition}
	The action profile by all the attackers $\vv = \left( v_1,v_2,...,v_m\right) \in \Vc^m$ is Nash Equilibrium (NE) of the game ${\cal G}_p$, with $p \in \{1,2,3\}$, if and only if it is a solution of the fix point equation
	\be
	\sum_{i \in \Kc}  v_i \ev_i \ev_i^{\sf T} = \textnormal{BR}^{(p)}\left(\vv\right),
	\ee
	with $\textnormal{BR}^{(p)}$: $\Vc^m \rightarrow \Sc_{+}^{m}$ being the global best response correspondence, that is,
	\begin{IEEEeqnarray}{rll}
		 \textnormal{BR}^{(p)}\left(\vv\right) = &\sum_{i \in \Kc}\textnormal{BR}_i^{(p)}\left(\vv \right)\ev_i \ev_i^{\sf T}. 
    \end{IEEEeqnarray} 
\end{definition}

Essentially, at an NE, attacker $i \in \Kc$ achieves the minimal cost given the actions adopted by all the other attackers.  
The following proposition highlights an important property of the game ${\cal G}_p$, with $p \in \{1,2,3\}$.
\begin{prop}
	For all $ p \in \{1,2,3\}$, the game $\gameNF_p$ possesses only one NE.
\end{prop}
\begin{proof}
	Note that for all $i \in \Kc$, the potential functions in~\eqref{potential_game1},~\eqref{potential_game2}, and~\eqref{potential_game3} are continuous over the set of all possible actions $v_i \in \Vc $. The set of actions $ \Vc $ is a convex set. From Proposition~\ref{prop_game1_cov}, the cost function for attacker $i$ is convex in each game. Hence, there is only one minimum for each potential function. From~\cite[Lemma 4.3]{potentialgame}, such a minimum corresponds to an NE. This completes the proof.
\end{proof}
The uniqueness of NE in a game guarantees the convergence of the actions by the attackers in the game.
 \subsection{Achievability of the NE}
The attackers are said to play a sequential best response dynamic (BRD) if the attackers can sequentially decide their own action, i.e., variance from their sets of best responses following a round-robin order. Let us denote the choice of attacker $i$ during round $t \in \mathbb{N}$ as $v_{i,t}^* $ and assume that attackers are able to observe all the other attackers' decisions. Under this assumption, the BRD is defined as follows.
\begin{definition}
	(Best Response Dynamics). The players of the game ${\cal G}_p$ are said to play a best response dynamics if there exists a round-robin order in which at each round $t \in \mathbb{N}$, the following holds
    \begin{IEEEeqnarray}{rll}
		 & v_{i,t}^*  \\
		 \nonumber
		 =  & \textnormal{BR}_i & \left(v_{1,t}^*,v_{1,t}^*,...,v_{i-1,t}^*,v_{i+1,t-1}^*,\ldots, v_{m-1,t-1}^*,v_{m,t-1}^*\right). 
    \end{IEEEeqnarray} 
\end{definition}
From the properties of potential games in~\cite[Lemma 4.2]{potentialgame}, the following lemma follows.
\begin{lemma}
	(Achievability of NE attacks). Any BRD in the game ${\cal G}_p$, with $p \in \{1,2,3\}$, converges to an attack construction that is the only NE in the game.
\end{lemma}

The proposed BRD corresponding to each ${\cal G}_p$, with $p \in \{1,2,3\}$, are described in
Algorithm~\ref{alg_game1}, Algorithm~\ref{alg_game2} and Algorithm~\ref{alg_game3}, respectively. 
\begin{algorithm}[H]
	\caption{Best Response Dynamics for $\gameNF_1$}\label{alg_game1} 
	\begin{algorithmic}
		\Require $\Hm$ in \eqref{eq:obs_noattack};\newline
		$\sigma^2$ in \eqref{EqZ};\newline
		$\Sigmam_{X\!X}$ in \eqref{Sigma_XX};\newline
		$\lambda \geq 1$ and initial actions $\vv_0 = \textbf{0}$.
		\Ensure actions at the NE $\vv^* \in \mathds{R}^m$.
		\For {$0<t<t_\textnormal{max}$},
		\For {$i \in \Kc$},
		\State Compute $v_{i,t}^*$ in~\eqref{eq_BR_game1} 
		\EndFor
		\State $t = t+1 $
		\EndFor
		\State Set $\vv^* = \left(v_1^*(t_\textnormal{max}),v_2^*(t_\textnormal{max}),\ldots,v_m^*(t_\textnormal{max})\right)$
	\end{algorithmic}
\end{algorithm}
 
\begin{algorithm}[H]
	\caption{Best Response Dynamics for $\gameNF_2$}\label{alg_game2} 
	\begin{algorithmic}
			\Require $\Hm$ in \eqref{eq:obs_noattack};\newline
	$\sigma^2$ in \eqref{EqZ};\newline
	$\Sigmam_{X\!X}$ in \eqref{Sigma_XX};\newline
	$\lambda \geq 0$ and 
	initial actions $\vv_0 = \textbf{0}$.
		\Ensure actions in the NE $\vv^* \in \mathds{R}^m$.
		\For {$0<t<t_\textnormal{max}$},
		\For {$i \in \Kc$},
		\State Compute $v_{i,t}^*$ in~\eqref{eq_BR_game2} 
		\EndFor
		\State $t = t+1 $
		\EndFor
		\State Set $\vv^* = \left(v_1^*(t_\textnormal{max}),v_2^*(t_\textnormal{max}),\ldots,v_m^*(t_\textnormal{max})\right)$
	\end{algorithmic}
\end{algorithm}
 
\begin{algorithm}[H]
	\caption{Best Response Dynamics for $\gameNF_3$}\label{alg_game3} 
	\begin{algorithmic}
		\Require $\Hm$ in \eqref{eq:obs_noattack};\newline
		$\sigma^2$ in \eqref{EqZ};\newline
		$\Sigmam_{X\!X}$ in \eqref{Sigma_XX};\newline
		$\lambda \geq 0$ and
		initial actions $\vv_0 = \textbf{0}$.
		\Ensure actions in the NE $\vv^*  \in \mathds{R}^m$.
		\For {$0<t<t_\textnormal{max}$},
		\For {$i \in \Kc$},
		\State Compute $v_{i,t}^*$ in~\eqref{eq_BR_game3} 
		\EndFor
		\State $t = t+1 $
		\EndFor
		\State Set $\vv^* = \left(v_1^*(t_\textnormal{max}),v_2^*(t_\textnormal{max}),\ldots,v_m^*(t_\textnormal{max})\right)$
	\end{algorithmic}
\end{algorithm}

The Jacobian matrix $\Hm \in \mathds{R}^{m\times n}$, system noise $\sigma^2 \in \mathds{R}_+$, second order moment of state variables $\Sigmam_{X\!X} \in \Sc_+^n$ and the weighting parameter $\lambda \geq 0$ are required in the proposed Algorithm~\ref{alg_game1}, Algorithm~\ref{alg_game2} and Algorithm~\ref{alg_game3}. We initialize the actions at $t = 0$ as $\vv_0 = \textbf{0}$. Within limited times, that is, while $t < t_{\max}$, player $i \in \Kc$ makes best response with respect to the actions by the other player $i \notin \Kc$ accordingly. The best response $v_i$ for $i \in \Kc$ is determined by the aim of the game, that is,~\eqref{eq_BR_game1},~\eqref{eq_BR_game2} and~\eqref{eq_BR_game3} for $\gameNF_1$, $\gameNF_2$ and $\gameNF_3$, respectively. In $\gameNF_1$, best response dynamic converges to the NE where the mutual information $I(X^n\|Y_A^m)$ and $D(P_{Y_A^m}\|P_{Y^m})$ are minimized. In $\gameNF_2$, best response dynamic converges to the NE where for all player $i \in \Kc$ the mutual information $I(X^n\|Y_{A_i})$ and $D(P_{Y_A^m}\|P_{Y^m})$ are minimized. In $\gameNF_3$, best response dynamic converges to the NE where for all player $i \in \Kc$ the mutual information $I(X^n\|Y_A^m)$ and $D(P_{Y_{A_i}}\|P_{Y_i})$ are minimized.

\section{Numerical Results}\label{sec_numerical_results}
In this section, the properties of the games $\Gc_p$, with $ p \in \{1,2,3\}$, described in Section~\ref{sec_game_formulation} are numerically evaluated on a direct current (DC) state estimation setting for IEEE test systems~\cite{UoW_ITC_99}. The voltage magnitude of the test systems are set to 1.0 per unit, that is the measurements in the systems are the power flow measurements between physically connected buses and power injection measurements that inject to the every buses. The Jacobian matrix $\Hm$ in~\eqref{eq:obs_noattack} is generated by MATPOWER~\cite{matpower}. An exponentially decaying Toeplitz matrix $\Sigmam_{X\!X} \in \Sc_+$ is adopted to capture the statistical dependence between the state variables where the strength of the correlation is set by a parameter $\rho$, that is, $(\Sigmam_{X\!X})_{ij} = \rho^{|i-j|}$ with $(i,j)\in \{1,2,\ldots,n\} \times \{1,2,\ldots,n\}$. Hence, apart from the correlation parameter $\rho$, the noise variance $\sigma^2$ in~\eqref{EqZ}, the Jacobian matrix $\Hm$ and the best response by the attackers, the performance of the game depends on the noise regime. In this setting, the noise regime of the observation model in terms of the signal to noise (SNR) is
\be
\textnormal{SNR} \eqdef 10\log_{10}\left(\frac{\textrm{tr}(\textbf{H}\Sigmam_{\textrm{XX}}\textbf{H}^\textrm{\sf T})}{m\sigma^{2}}\right).
\ee 
In the simulation, note that the game starts with~$\vv = \textbf{0}$, that is, all the attackers have not attack the system in $t = 0$.
\subsection{Performance in terms of the potentials}
Fig.~\ref{game1_bus9_potential} depicts the potential function of $\Gc_1$ given by~\eqref{game1_cost} in terms of round robin for different $\lambda$ when $\rho = 0.9$, SNR = 30 dB in the IEEE 9-bus test system. The NE equilibria has been numerically evaluated and is presented by a red square. In $\Gc_1$, the potential function is the same as the cost function of attacker $i, i \in \{1,2,\ldots,m\}$. From $t = 0$ to $t = 1$, the potential function decreases monotonically, which implies that all attackers in the system benefit from the attacks launched by the other attackers in the round robin. Note that after all the attacker have injected an attack, that is, $t = 1$, the potential keeps on decreasing until the NE equilibria is achieved. Fig.~\ref{game2_bus9_potential} and Fig.~\ref{game3_bus9_potential} depict the potential function of $\Gc_2$ given by~\eqref{game2_cost} and $\Gc_3$ given by~\eqref{game3_cost}, respectively, with the same setting as in Fig.~\ref{game1_bus9_potential}. Surprisingly, in $\Gc_2$, the potential increases from $t = 0$ to $t = 1$ as during the first round robin, the attackers compromise the measurements for the first time to decrease the local information defined in~\eqref{eq_local_I}. However, the attacks in the first round robin significantly increase the KL divergence that results in increasing potential. From $t > 1$, the attackers modify their own action to optimize the cost of the attacks defined in~\eqref{game2_cost} and reach to the NE equilibria that is represented by red square.
\begin{figure}[htbp]
	\centering
	\begin{minipage}[t]{0.49\textwidth}
		\includegraphics[width=8.5cm]{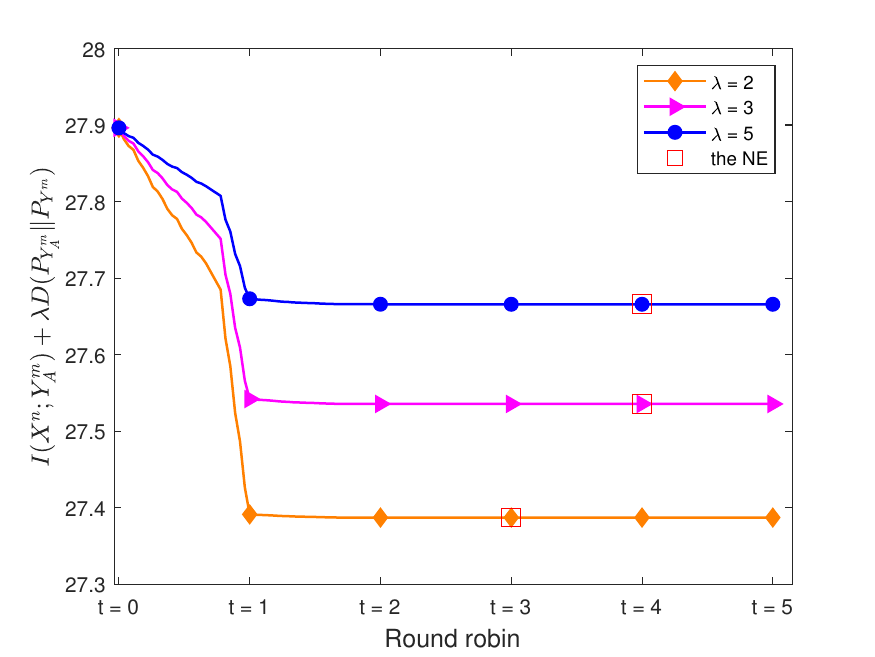}
		\caption{The convergence in {$\Gc_1$} in terms of the potential function $P_1$ on IEEE 9-bus test system with $\rho=0.9$, SNR = 30 dB. The red squares show the potential in the NEs.}\label{game1_bus9_potential} 
	\end{minipage}
	\centering
	\begin{minipage}[t]{0.49\textwidth}
		\includegraphics[width=8.5cm]{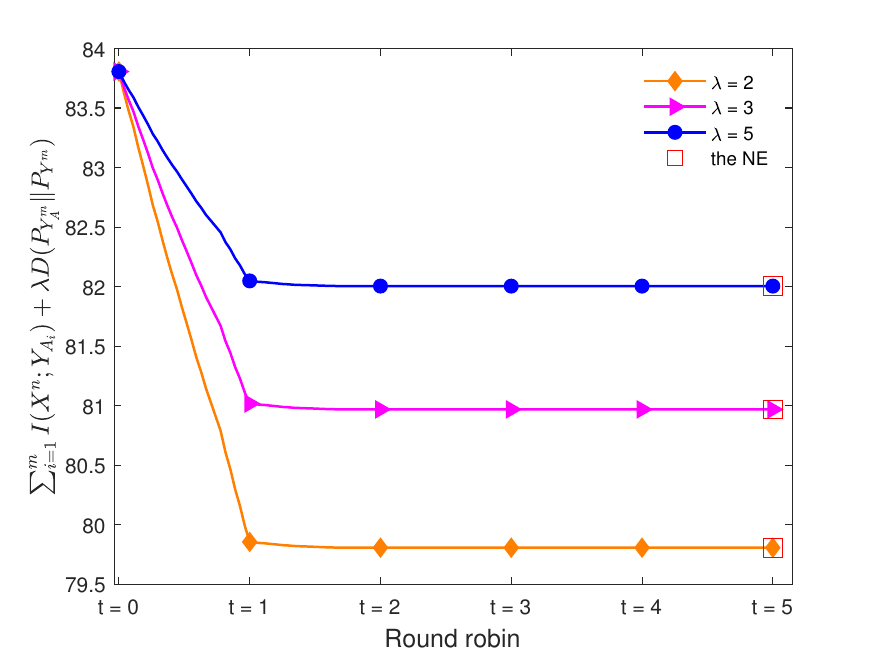}
		\caption{The convergence in {$\Gc_2$} in terms of the potential function $P_2$ on IEEE 9-bus test system with $\rho=0.9$, SNR = 30 dB. The red squares show the potential in the NEs.}\label{game2_bus9_potential} 
	\end{minipage}
	\centering
	\begin{minipage}[t]{0.49\textwidth}
		\includegraphics[width=8.5cm]{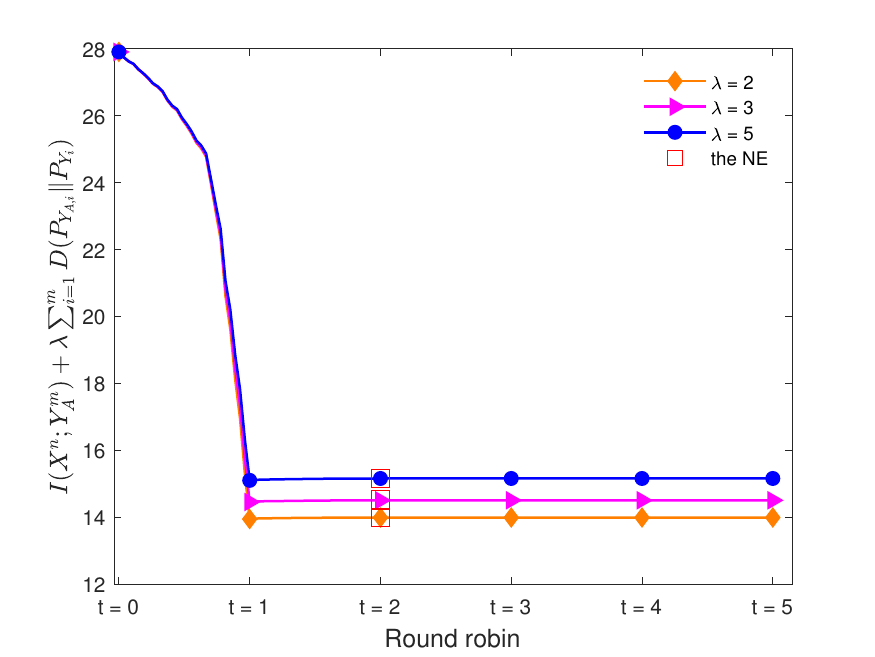}
		\caption{The convergence in {$\Gc_3$} in terms of the potential function $P_3$ on IEEE 9-bus test system with $\rho=0.9$, SNR = 30 dB. The red squares show the potential in the NEs.}\label{game3_bus9_potential} 
	\end{minipage}
\end{figure}
 
\begin{figure}[htbp]
 	\centering
\begin{minipage}[t]{0.49\textwidth}
	\includegraphics[width=8.5cm]{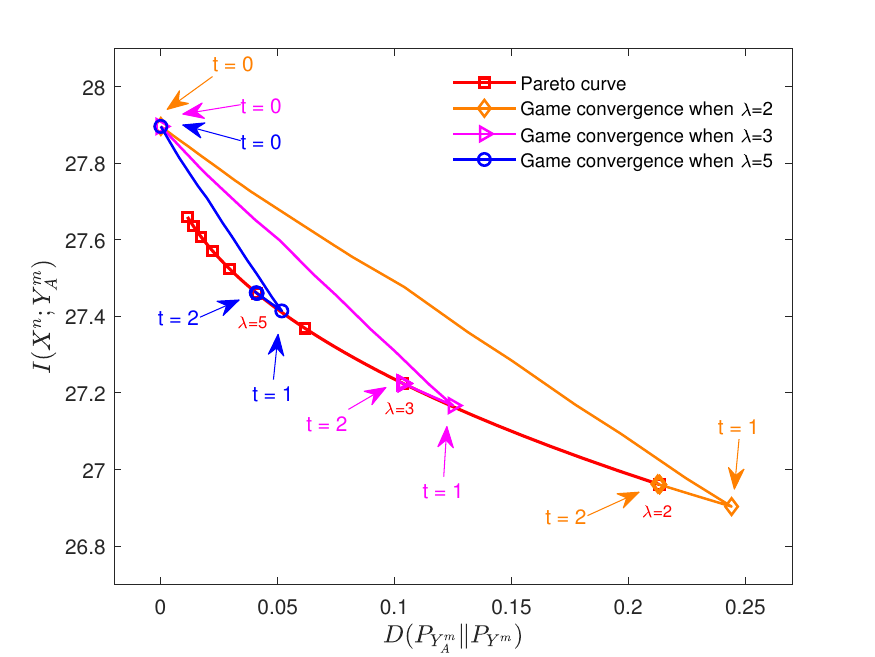}
	\caption{The convergence in {$\Gc_1$} in terms of tradeoff between mutual information and KL divergence on IEEE 9-bus test system with $\rho=0.9$, SNR = 30 dB.}\label{game1_bus9_tradeoff} 
\end{minipage}
\centering
\begin{minipage}[t]{0.49\textwidth}
	\includegraphics[width=8.5cm]{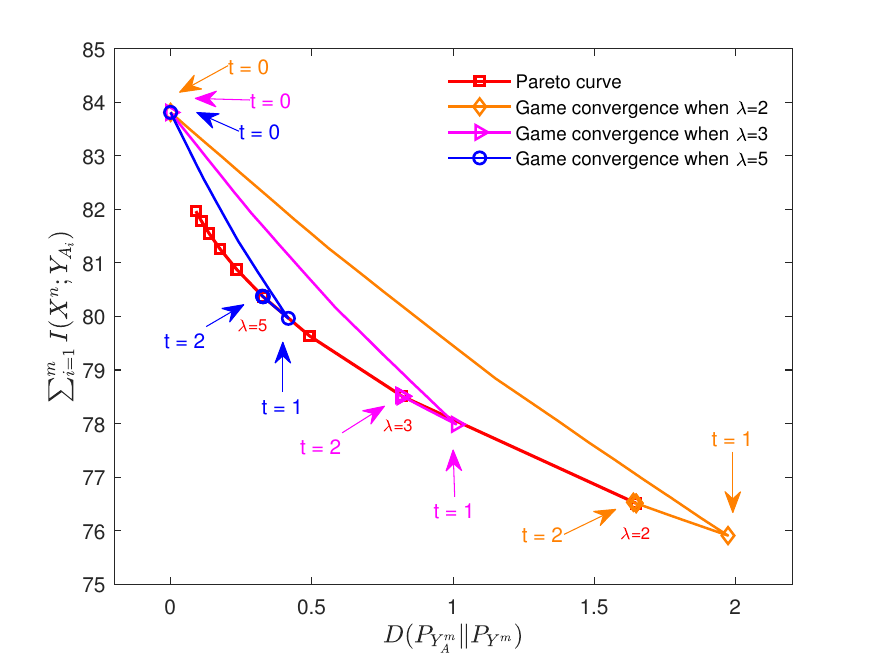}
	\caption{The convergence in {$\Gc_2$} in terms of tradeoff between mutual information and KL divergence on IEEE 9-bus test system with $\rho=0.9$, SNR = 30 dB.}\label{game2_bus9_tradeoff} 
\end{minipage}
\centering
\begin{minipage}[t]{0.49\textwidth}
	\includegraphics[width=8.5cm]{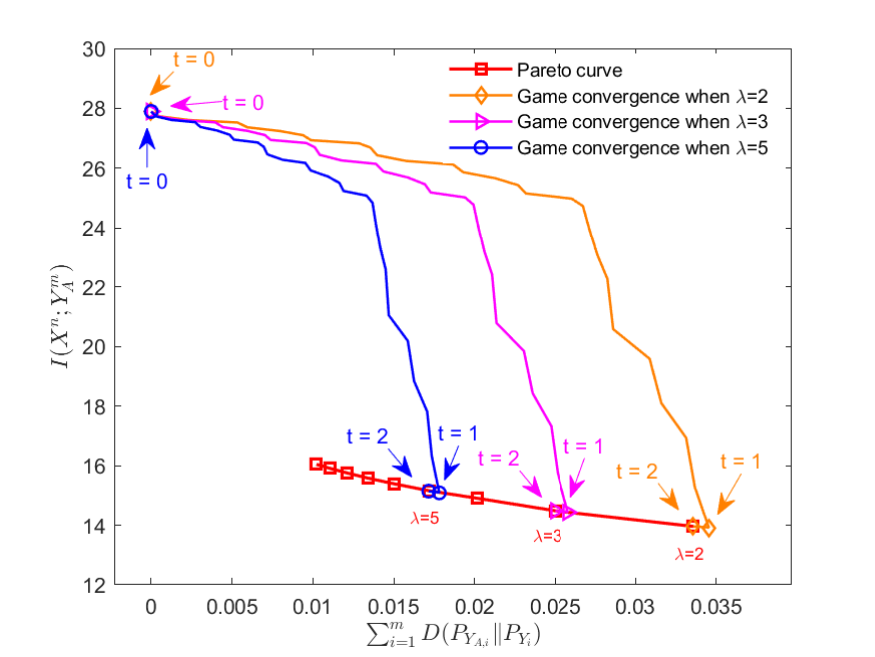}
	\caption{The convergence in {$\Gc_3$} in terms of tradeoff between mutual information and KL divergence on IEEE 9-bus test system with $\rho=0.9$, SNR = 30 dB.}\label{game3_bus9_tradeoff} 
\end{minipage}
\end{figure}

\subsection{Performance in terms of the tradeoff in the game} 
Fig.~\ref{game1_bus9_tradeoff} depicts the convergence in $\Gc_1$ in terms of the tradeoff between mutual information and KL divergence for different $\lambda$ when $\rho=0.9$, SNR = 30 dB on IEEE 9-bus system. From $t = 0$ to $t = 1$, the attackers compromises the corresponding measurements in order which causes mutual information decrease as well as the increasing of KL divergence. Interestingly, from $t = 1$ to $t = 2$, the mutual information increases while KL divergence decreases. It is worth noting that this is because in $t = 0$, there was not attack in the system which allows the attacker to launch an attack that is under the KL divergence constraints easily. Therefore, when all the measurements being compromised in $t = 1$, the overall attacks do not meet the KL divergence constraints. Consequently, from $t = 1$ to $t = 2$, the attackers modify the corresponding attacks and eventually converge to the pareto curve. Fig.~\ref{game2_bus9_tradeoff} and Fig.~\ref{game3_bus9_tradeoff} depict the convergence in $\Gc_2$ and $\Gc_3$ with the same setting as in Fig.~\ref{game1_bus9_tradeoff}. We observe the same phenomenon that in $t = 1$, the attacks do not meet the KL divergence constraints and results in an increasing mutual information and decreasing KL divergence from $t = 1$ to $t = 2$. Interestingly, in Fig.~\ref{game3_bus9_tradeoff}, the mutual information comes across a significant decrease in the later of $t = 0$ to $t = 1$. It is worth noting that the significant decrease comes from the attacks injected to power injection measurements. In~\cite{YE_IETSG_22}, the power injection measurements are more vulnerable to data integrity attacks than power flow measurements.

\section{Conclusion}\label{sec_conclusion}
We have proposed a novel decentralized stealth attack construction that targets at CPS. The objectives of the attack constructions are the disruption and detection both globally and locally that are measured by information metrics. We utilize the interaction between the attackers to formulate games in attack constructions. We have proved the existence and convexity of the potentials in different games that are motived by different information metrics. We proposed best response dynamics to achieve the Nash Equilibrium of the games, accordingly. We have numerically evaluated the performance of the decentralized attacks in IEEE test systems and shown the interaction between the attackers converge to the Nash Equilibrium.

%%\section*{Acknowledgment}
\bibliographystyle{IEEEtran}
\bibliography{fullgamebib}
\onecolumn
\section*{Appendix A}\label{proof_prop_local_I}

\section*{Proof of Proposition~\ref{prop_local_I}}

\begin{proof}
	Note that from~\eqref{ith_noatt} and~\eqref{ith_compromised}, it follows that 
	\be
	Y_i \sim {\cal N}\left(0, \ev_i^{\sf T}   \Sigmam_{Y\!Y} \ev_i \right),
	\ee
	and
	\be
	Y_{A,i} \sim  {\cal N}\left(0, \ev_i^{\sf T}   \Sigmam_{Y\!Y} \ev_i + v_i \right).
	\ee 
	
	The mutual information between random vector of state variables $X^n \sim {\cal N}(\textbf{0}, \Sigmam_{X\!X})$ and the $i$-th random measurement $Y_{A,i}$ is
	\be
	I(X^n; Y_{A,i})  \eqdef  \mathds{E}_{X^n,Y_{A,i}} \left[\log  \dfrac{f_{X^n,Y_{A,i}}}{f_{X^n}f_{Y_{A,i}}}      \right],
	\ee
	where $f_{X^n}, f_{Y_{A,i}}$ and $f_{X^n,Y_{A,i}}$ are the PDF of the random variables $X^n$, $Y_{A_i}$ and $(X^n,Y_{A_i})$, respectively. 
	Let us denote the random vector $W^{n+1} \eqdef (X^n,  Y_{A,i})\in \mathds{R}^{n+1}$ and $f_W$ be the PDF of $W$ such that $f_W = f_{X^n,Y_{A,i}}$. The random vector $W^{n+1}$ follows a joint multivariate Gaussian distribution given by
	\be
	W^{n+1} \sim {\cal N}(\textbf{0}, \Sigmam),
	\ee
	where the block covariance matrix has the following structure:
	\be
	\Sigmam \eqdef 
	\begin{bmatrix}
		\Sigmam_{X\!X} & \Sigmam_{X\!X}	\hv_i^{\sf T}  \\
		\hv_i \Sigmam_{X\!X}  &  \ev_i^{\sf T} \Sigmam_{Y\!Y} \ev_i+v_i\\
	\end{bmatrix},
	\ee
	where $\hv_i$ is the $i$-th row of $\Hm$. Hence, the following holds
	\begin{IEEEeqnarray}{rll}
		\nonumber
		&I(X^n; Y_{A,i}) \\
		= & \mathds{E}_{W} \left[\log   \dfrac{f_W}{f_{X^n}f_{Y_{A,i}}}      \right]    \\
		= & \mathds{E}_{W} \left[\log    f_W      \right] - \mathds{E}_{W} \left[\log  f_{X^n} \right] - \mathds{E}_{W} \left[\log  f_{Y_{A,i}} \right]   \label{20220920_6} \\
		= & \mathds{E}_{W} \left[\textrm{log}      \frac{\textrm{exp}\left(-\frac{1}{2}W^{\sf T} \Sigmam^{-1} W \right)}{   (2 \pi)^{\frac{n+1}{2}} |\Sigmam|^{\frac{1}{2}}}  \right]  -  \mathds{E}_X \left[\textrm{log}  \frac{ \textrm{exp}\left(-\frac{1}{2}X^{\sf T} \Sigmam_{X\!X}^{-1} X \right)}{   (2 \pi)^{\frac{n}{2}} |\Sigmam_{X\!X}|^{\frac{1}{2}}}    \right]   \label{20220920_5}  - \mathds{E}_{Y_{A, i}} \left[\textrm{log}  \frac{ \textrm{exp}\left(-\frac{1}{2} \frac{y_a^2}{  \ev_i^{\sf T}   \Sigmam_{Y\!Y} \ev_i + v_i    }   \right)}{  (2 \pi)^{\frac{1}{2}} (\ev_i^{\sf T}   \Sigmam_{Y\!Y} \ev_i + v_i)^{\frac{1}{2}} }    \right] \\
		= &\dfrac{1}{2} \mathds{E}_{W}\left[- W^{\sf T}\Sigmam^{-1}W - \textrm{log}| \Sigmam |\right]  +  \dfrac{1}{2} \mathds{E}_{X}\left[ X^{\sf T}\Sigmam_{X\!X}^{-1}X +   \textrm{log}| \Sigmam_{X\!X} |\right]  \label{20220920_7} \\
		\nonumber
		&+  \dfrac{1}{2} \mathds{E}_{Y_{A,i}}\left[ \frac{y_a^2}{  \ev_i^{\sf T}   \Sigmam_{Y\!Y} \ev_i + v_i    }    +   \textrm{log}(\ev_i^{\sf T}   \Sigmam_{Y\!Y} \ev_i + v_i) \right]               \\
		= &  \dfrac{1}{2}  \mathds{E}_{W}\!\left[- \!\textrm{tr}(\Sigmam^{-1}\! W \! W^{\sf T})  \right] \! \!+ \! \dfrac{1}{2}  \mathds{E}_{X}\left[ \textrm{tr}(\Sigmam_{X\!X}^{-1}\! X\! X^{\sf T})\right]     + \dfrac{1}{2\left(\ev_i^{\sf T}   \Sigmam_{Y\!Y} \ev_i + v_i\right)}\mathds{E}_{Y_{A,i}}\left[    y_a^2      \right]    \\
		\nonumber
		& +  \dfrac{1}{2}  \textrm{log}       \dfrac{ |\Sigmam_{X\!X}| (  \ev_i^{\sf T}   \Sigmam_{Y\!Y} \ev_i + v_i       )}{|\Sigmam|} \\
		= &   \dfrac{1}{2} \left( - (n+1) + n + 1 \right)  + \dfrac{1}{2}  \textrm{log}       \dfrac{ |\Sigmam_{X\!X}|(  \ev_i^{\sf T}   \Sigmam_{Y\!Y} \ev_i + v_i       )}{|\Sigmam|}         \\
		= &\dfrac{1}{2}  \textrm{log}       \dfrac{ |\Sigmam_{X\!X}| (  \ev_i^{\sf T}   \Sigmam_{Y\!Y} \ev_i + v_i       )}{|\Sigmam|},  \label{20220920_8}
	\end{IEEEeqnarray}
	where the equality in~\eqref{20220920_5} holds from taking the probability density functions of $W^{n+1}$, $Y_{A,i}$ and $X^n$ to~\eqref{20220920_6} and the equality in~\eqref{20220920_7} follows from taking the constants out of the expectation.
	
	The following also holds
	\begin{IEEEeqnarray}{rll}
		\label{20220518_1}
		\nonumber
		& | \Sigmam | \\
		=& |\Sigmam_{X\!X}|| \ev_i^{\sf T}   \Sigmam_{Y\!Y} \ev_i +v_i - \hv_i \Sigmam_{X\!X} \Sigmam_{X\!X}^{-1}  \Sigmam_{X\!X}	\hv_i^{\sf T}| \label{20220920_2}\\
		=& |\Sigmam_{X\!X}|| \ev_i^{\sf T}  \Hm\Sigmam_{X\!X}\Hm^{\sf T} \ev_i +\sigma^2 + v_i - \hv_i \Sigmam_{X\!X} 	\hv_i^{\sf T}|  \label{20220920_3}\\
		=& |\Sigmam_{X\!X}|( \sigma^2 + v_i  ),  \label{20220920_4}
	\end{IEEEeqnarray}
	where the equality in~\eqref{20220920_2} holds from~\cite[14.17(a)]{seber}; the equality in~\eqref{20220920_3} follows from $\Sigmam_{Y\!Y} = \Hm\Sigmam_{X\!X}\Hm^{\sf T} + \sigma^2\textbf{I}_m$ and the equality in~\eqref{20220920_4} follows from $\hv_i = \ev_i^{\sf T}  \Hm$ and $\hv_i^{\sf T} =  \Hm^{\sf T} \ev_i$.
	Therefore,~\eqref{20220920_8} yields that 
	\begin{IEEEeqnarray}{rll}
		I(X^n; Y_i + A_i)  =  &\dfrac{1}{2}  \textrm{log}       \dfrac{ |\Sigmam_{X\!X}| (  \ev_i^{\sf T}   \Sigmam_{Y\!Y} \ev_i + v_i       )}{|\Sigmam|}   \\
		=  &\dfrac{1}{2}  \textrm{log}       \dfrac{ |\Sigmam_{X\!X}| (  \ev_i^{\sf T}   \Sigmam_{Y\!Y} \ev_i + v_i       )}{|\Sigmam_{X\!X}|( \sigma^2 + v_i  )}   \\
		% 		=  &\dfrac{1}{2}  \textrm{log}       \dfrac{     \ev_i^{\sf T}   \Sigmam_{Y\!Y} \ev_i + v_i        }{   \sigma^2 + v_i  }\\
		= & \frac{1}{2}\log \left(1+\frac{ \ev_i^{\sf T}\Hm \Sigmam_{X\!X} \Hm^{\sf T}\ev_i }{\sigma^2 +v_i}\right).   
	\end{IEEEeqnarray}
	This completes the proof.
	
\end{proof}

\section*{Appendix B}\label{app_B}
\section*{Proof of Proposition~\ref{prop_local_D}}\label{proof_prop_local_D}

\begin{proof}
	
	Let $f_{P_{Y_{A, i}}}$ and $f_{P_{Y_i}}$ denote the PDF of $P_{Y_{A, i}}$ and $P_{Y_i} $, respectively. Note that $Y_i \sim {\cal N}\left(0,\textnormal{tr}\left(\Hm \Sigmam_{X\!X} \Hm^{\sf T}\ev_i \ev_i^{\sf T}\right) + \sigma^2\right)$ and $Y_{A,i}~\sim~{\cal N}\left(0,\textnormal{tr}\left(\Hm \Sigmam_{X\!X} \Hm^{\sf T}\ev_i \ev_i^{\sf T}\right) + \sigma^2 + v_i\right)$. The KL divergence between $P_{Y_{A, i}}$ and $P_{Y_i}$ is given by
	
	\begin{IEEEeqnarray}{rll}
		&D(P_{Y_{A, i}} \| P_{Y_i} ) \\
		\eqdef & \mathds{E}_{P_{Y_{A,i}}} \left[\textrm{log} \dfrac{f_{Y_{A, i} }}{f_{Y_i}} \right]  \\
		= & \mathds{E}_{P_{Y_{A,i}}} \left[\textrm{log} \frac{\dfrac{1}{\sqrt{\textnormal{tr}\left(\Hm \Sigmam_{X\!X} \Hm^{\sf T}\ev_i \ev_i^{\sf T}\right)+\sigma^2 +v_i}\sqrt{2\pi}  }\textnormal{exp}\left[ \frac{-x^2}{2(\textnormal{tr}\left(\Hm \Sigmam_{X\!X} \Hm^{\sf T}\ev_i \ev_i^{\sf T}\right)+\sigma^2 +v_i)}\right] }{\frac{1}{\sqrt{\textnormal{tr}\left(\Hm \Sigmam_{X\!X} \Hm^{\sf T}\ev_i \ev_i^{\sf T}\right)+\sigma^2}\sqrt{2\pi}}\textnormal{exp}\left[ \frac{-x^2}{2\textnormal{tr}\left(\Hm \Sigmam_{X\!X} \Hm^{\sf T}\ev_i \ev_i^{\sf T}\right) +\sigma^2}\right]}\right] \label{20211025_1}\\
		=& \dfrac{1}{2}  \mathds{E}_{P_{Y_{A,i}}} \left[ \dfrac{-x^2}{  \textnormal{tr}\left(\Hm \Sigmam_{X\!X} \Hm^{\sf T}\ev_i \ev_i^{\sf T}\right)+\sigma^2+v_i }   -   \dfrac{-x^2}{  \textnormal{tr}\left(\Hm \Sigmam_{X\!X} \Hm^{\sf T}\ev_i \ev_i^{\sf T}\right)+\sigma^2 }   \right]  + \dfrac{1}{2} \textrm{log} \dfrac{\textnormal{tr}\left(\Hm \Sigmam_{X\!X} \Hm^{\sf T}\ev_i \ev_i^{\sf T}\right)+\sigma^2}{ \textnormal{tr}\left(\Hm \Sigmam_{X\!X} \Hm^{\sf T}\ev_i \ev_i^{\sf T}\right) +\sigma^2+ v_i    }   \\
		=&\dfrac{1}{2} \dfrac{v_i}{\left(\textnormal{tr}\left(\Hm \Sigmam_{X\!X} \Hm^{\sf T}\ev_i \ev_i^{\sf T}\right)+\sigma^2\right)\left(\left(\textnormal{tr}\left(\Hm \Sigmam_{X\!X} \Hm^{\sf T}\ev_i \ev_i^{\sf T}\right)+\sigma^2+v_i\right)\right)} \mathds{E}_{P_{Y_{A,i}}} \left[ x^2  \right]\\
		\nonumber
		&+ \dfrac{1}{2} \textrm{log} \dfrac{\textnormal{tr}\left(\Hm \Sigmam_{X\!X} \Hm^{\sf T}\ev_i \ev_i^{\sf T}\right)+\sigma^2}{ \textnormal{tr}\left(\Hm \Sigmam_{X\!X} \Hm^{\sf T}\ev_i \ev_i^{\sf T}\right) +\sigma^2+ v_i    }    \\
		\nonumber
		=&\dfrac{1}{2} \dfrac{v_i}{\left(\textnormal{tr}\left(\Hm \Sigmam_{X\!X} \Hm^{\sf T}\ev_i \ev_i^{\sf T}\right)+\sigma^2\right)\left(\textnormal{tr}\left(\Hm \Sigmam_{X\!X} \Hm^{\sf T}\ev_i \ev_i^{\sf T}\right)+\sigma^2+v_i\right)}\left(\textnormal{tr}\left(\Hm \Sigmam_{X\!X} \Hm^{\sf T}\ev_i \ev_i^{\sf T}\right)+\sigma^2+v_i\right) \\
		&+ \dfrac{1}{2} \textrm{log} \dfrac{\textnormal{tr}\left(\Hm \Sigmam_{X\!X} \Hm^{\sf T}\ev_i \ev_i^{\sf T}\right)+\sigma^2}{ \textnormal{tr}\left(\Hm \Sigmam_{X\!X} \Hm^{\sf T}\ev_i \ev_i^{\sf T}\right) +\sigma^2+ v_i    }  \label{20211025_2} \\
		=& \dfrac{1}{2}\left(   \dfrac{v_i}{\textnormal{tr}\left(\Hm \Sigmam_{X\!X} \Hm^{\sf T}\ev_i \ev_i^{\sf T}\right)+\sigma^2 } + \log \dfrac{\textnormal{tr}\left(\Hm \Sigmam_{X\!X} \Hm^{\sf T}\ev_i \ev_i^{\sf T}\right)+\sigma^2}{ \textnormal{tr}\left(\Hm \Sigmam_{X\!X} \Hm^{\sf T}\ev_i \ev_i^{\sf T}\right) +\sigma^2+ v_i    } \right), 
	\end{IEEEeqnarray} 
	where $\eqref{20211025_1}$ follows from taking the density function of $P_{Y_{A, i}}$ and $P_{Y_i} $; $\eqref{20211025_2}$ follows from the fact that the expectation of the random variable $X^2$ such that $X \sim {\cal N}(0,\textnormal{tr}\left(\Hm \Sigmam_{X\!X} \Hm^{\sf T}\ev_i \ev_i^{\sf T}\right)+\sigma^2 + v_i )$ is $\textnormal{tr}\left(\Hm \Sigmam_{X\!X} \Hm^{\sf T}\ev_i \ev_i^{\sf T}\right) +\sigma^2+ v_i$. 
	This completes the proof.
	
\end{proof}

\section*{Appendix C}
\section*{Proof of Proposition~\ref{prop_game1_cov}}\label{proof_prop_game3_cvx}

\begin{proof}
	Let us define
	\be
	\Am~\eqdef~\Hm\Sigmam_{X\!X}\Hm^{\sf T}~\sum_{j \in {\cal K} \setminus \{i\}} \dfrac{1}{\sigma^2 + v_j} \ev_j \ev_j^{\sf T}    + \textbf{I}_m,
	\ee
	and $a_i \eqdef  \textnormal{tr}\left(  \Am^{-1}\Hm\Sigmam_{X\!X}\Hm^{\sf T}\ev_i \ev_i^{\sf T} \right)$. The derivative of the term $\textnormal{log} \left|\dfrac{1}{\sigma^2 + v_i} \Hm\Sigmam_{X\!X}\Hm^{\sf T}\ev_i \ev_i^{\sf T} + \Am\right|$ with respect to $v_i$ is 
	
	\begin{IEEEeqnarray}{rll} 
		& \dfrac{\partial }{\partial v_i}\textnormal{log} \left|\dfrac{1}{\sigma^2 + v_i} \Hm\Sigmam_{X\!X}\Hm^{\sf T}\ev_i \ev_i^{\sf T} + \Am\right| \\
		=&-\dfrac{1}{ (\sigma^2 + v_i)^2} \textnormal{tr}\left( \left(\dfrac{1}{\sigma^2 + v_i} \Hm\Sigmam_{X\!X}\Hm^{\sf T}\ev_i \ev_i^{\sf T} +  \Am \right)^{-1} \Hm\Sigmam_{X\!X}\Hm^{\sf T}\ev_i \ev_i^{\sf T}\right) \label{20211001_1} \\
		= &-\dfrac{1}{ (\sigma^2 + v_i)^2} \textnormal{tr}\left( \left(\Am^{-1} - \dfrac{\dfrac{1}{  \sigma^2 + v_i }}{1+\dfrac{1}{  \sigma^2 + v_i } \ev_i^{\sf T} \Am^{-1} \Hm\Sigmam_{X\!X}\Hm^{\sf T}\ev_i} \Am^{-1}\Hm\Sigmam_{X\!X}\Hm^{\sf T}\ev_i \ev_i^{\sf T}\Am^{-1}  \right) \Hm\Sigmam_{X\!X}\Hm^{\sf T}\ev_i \ev_i^{\sf T}\right) \label{20211001_2} \\
		= &-\dfrac{1}{ (\sigma^2 + v_i)^2} \textnormal{tr}\left(  \Am^{-1}\Hm\Sigmam_{X\!X}\Hm^{\sf T}\ev_i \ev_i^{\sf T} \right) \\
		&+\dfrac{1}{ (\sigma^2 + v_i)^2}\dfrac{\dfrac{1}{   \sigma^2 + v_i}}{1+\dfrac{1}{ \sigma^2 + v_i} \ev_i^{\sf T} \Am^{-1} \Hm\Sigmam_{X\!X}\Hm^{\sf T}\ev_i} \textnormal{tr}\left(   \Am^{-1}\Hm\Sigmam_{X\!X}\Hm^{\sf T}\ev_i \ev_i^{\sf T}\Am^{-1} \Hm\Sigmam_{X\!X}\Hm^{\sf T}\ev_i \ev_i^{\sf T}   \right) \label{20211001_3} \\
		= & -\dfrac{1}{ (\sigma^2 + v_i)^2}\alpha_i + \dfrac{1}{ (\sigma^2 + v_i)^2}\dfrac{\dfrac{1}{  \sigma^2 + v_i }}{1+\dfrac{1}{ \sigma^2 + v_i  } \alpha_i}  \alpha_i^2 \label{20211021_1} \\
		= & - \dfrac{\alpha_i}{(\sigma^2+v_i)\left(\sigma^2+v_i+\alpha_i\right)} \label{20220519_8},
	\end{IEEEeqnarray} 
	where the equality in~\eqref{20211001_1} follows from taking the derivative of $\textnormal{log} \left|\dfrac{1}{\sigma^2+ v_i} \Hm\Sigmam_{X\!X}\Hm^{\sf T}\ev_i \ev_i^{\sf T} +  \Am \right|$ with respect of $v_i$~\cite[Statement 17.18(a)]{seber} and taking $-\dfrac{1}{(\sigma^2+ v_i)^2}$ out of the trace, \eqref{20211001_2} follows from Sherman-Morrison Formula in~\cite[15.2(b)]{seber}. Note that $\alpha_i>0$ and $v_i>0$. It follows that $- \dfrac{\alpha_i}{(\sigma^2+v_i)\left(\sigma^2+v_i+\alpha_i\right)} < 0$.
	We now proceed to  exam the second derivative of $\textnormal{log} \left|  \dfrac{1}{\sigma^2+ v_i} \Hm\Sigmam_{X\!X}\Hm^{\sf T}\ev_i \ev_i^{\sf T} + \Am \right|$, that is,
	\begin{IEEEeqnarray}{rll} 
		&\dfrac{\partial^2 }{\partial v_i^2} \left(- \dfrac{\alpha_i}{(\sigma^2+v_i)\left(\sigma^2+v_i+\alpha_i\right)}\right) 
		= \dfrac{\alpha_i\left(2(\sigma^2+v_i)+\alpha_i\right)}{(\sigma^2+v_i)^2\left(\sigma^2+v_i+\alpha_i\right)^2}. 
	\end{IEEEeqnarray} 
	This completes the proof.
	
\end{proof}
 
\end{document}

%% file: macros.tex
\long\def\comment#1{}

\newcommand{\be}{\begin{equation}}
\newcommand{\ee}{\end{equation}}
\usepackage{amsthm}
\newtheorem{theorem}{Theorem}

\newtheorem{lemma}[theorem]{Lemma}
\newtheorem{definition}{Definition}
 
\newtheorem{prop}{Proposition}
\newtheorem*{proof*}{Proof}
 
\newfont{\bbb}{msbm10 scaled 700}

\newfont{\bb}{msbm10 scaled 1100}

\newcommand{\ev}{{\bf e}}

\newcommand{\hv}{{\bf h}}

\newcommand{\vv}{{\bf v}}
\newcommand{\xv}{{\bf x}}
\newcommand{\yv}{{\bf y}}

\newcommand{\Am}{{\bf A}}

\newcommand{\Hm}{{\bf H}}

\newcommand{\Sm}{{\bf S}}

\newcommand{\Ac}{{\cal A}}

\newcommand{\Gc}{{\cal G}}
\newcommand{\Hc}{{\cal H}}

\newcommand{\Kc}{{\cal K}}

\newcommand{\Mc}{{\cal M}}
\newcommand{\Nc}{{\cal N}}

\newcommand{\Sc}{{\cal S}}

\newcommand{\Vc}{{\cal V}}

\newcommand{\RNum}[1]{\uppercase\expandafter{\romannumeral #1\relax}}

\newcommand{\muv}{\hbox{\boldmath$\mu$}}

\newcommand{\Sigmam}{\hbox{\boldmath$\Sigma$}}

\newcommand{\eqdef}{\stackrel{\Delta}{=}}

\newcommand{\gameNF}{\mathcal{G}}

\usepackage{dsfont}
\usepackage{amsmath}
\DeclareMathOperator*{\argmin}{arg\,min}

\usepackage[margin=1in]{geometry}
\usepackage{hyperref}
\usepackage{circuitikz}
\usepackage{booktabs}
\usepackage{amssymb}
\usepackage{algorithm}
\usepackage{algpseudocode}

\usepackage{graphicx}
\usepackage{epstopdf}
\usepackage{multirow}
\usepackage{cite}

\usepackage{mathrsfs}

%% file: Manuscript.bbl
\begin{thebibliography}{10}
\providecommand{\url}[1]{#1}
\csname url@samestyle\endcsname
\providecommand{\newblock}{\relax}
\providecommand{\bibinfo}[2]{#2}
\providecommand{\BIBentrySTDinterwordspacing}{\spaceskip=0pt\relax}
\providecommand{\BIBentryALTinterwordstretchfactor}{4}
\providecommand{\BIBentryALTinterwordspacing}{\spaceskip=\fontdimen2\font plus
\BIBentryALTinterwordstretchfactor\fontdimen3\font minus
  \fontdimen4\font\relax}
\providecommand{\BIBforeignlanguage}[2]{{%
\expandafter\ifx\csname l@#1\endcsname\relax
\typeout{** WARNING: IEEEtran.bst: No hyphenation pattern has been}%
\typeout{** loaded for the language `#1'. Using the pattern for}%
\typeout{** the default language instead.}%
\else
\language=\csname l@#1\endcsname
\fi
#2}}
\providecommand{\BIBdecl}{\relax}
\BIBdecl

\bibitem{NSF_web}
\BIBentryALTinterwordspacing
N.~S. Foundation. Cyber-physical systems: Enabling a smart and connected world.
  [Online]. Available:
  \url{https://www.nsf.gov/news/special_reports/cyber-physical/}
\BIBentrySTDinterwordspacing

\bibitem{Advisor_UK_24}
\BIBentryALTinterwordspacing
P.~C. of~Advisors~on Science and Technology. (2024) Strategy for cyber-physical
  resilience: Fortifying our critical infrastructure for a digital world.
  [Online]. Available:
  \url{https://www.whitehouse.gov/wp-content/uploads/2024/02/PCAST_Cyber-Physical-Resilience-Report_Feb2024.pdf}
\BIBentrySTDinterwordspacing

\bibitem{AR_TonEmergingTopicsinComputing_17}
B.~Aksanli and T.~S. Rosing, ``Human behavior aware energy management in
  residential cyber-physical systems,'' \emph{IEEE Trans. on Emerging Topics in
  Computing}, vol.~8, no.~1, pp. 45--57, Mar. 2017.

\bibitem{HZ_TVT_15}
Y.~Hou, Y.~Zhao, A.~Wagh, L.~Zhang, C.~Qiao, K.~F. Hulme, C.~Wu, A.~W. Sadek,
  and X.~Liu, ``Simulation-based testing and evaluation tools for
  transportation cyber--physical systems,'' \emph{IEEE Trans. on Veh.
  Technol.}, vol.~65, no.~3, pp. 1098--1108, Feb. 2015.

\bibitem{SK_ProceedingsofIEEE_13}
K.~Sampigethaya and R.~Poovendran, ``Aviation cyber--physical systems:
  Foundations for future aircraft and air transport,'' \emph{Proceedings of the
  IEEE}, vol. 101, no.~8, pp. 1834--1855, Mar. 2013.

\bibitem{MK_CirpAnnals_16}
L.~Monostori, B.~K{\'a}d{\'a}r, T.~Bauernhansl, S.~Kondoh, S.~Kumara,
  G.~Reinhart, O.~Sauer, G.~Schuh, W.~Sihn, and K.~Ueda, ``Cyber-physical
  systems in manufacturing,'' \emph{{CIRP} {A}nnals}, vol.~65, no.~2, pp.
  621--641, Aug. 2016.

\bibitem{LF_Industry4_14}
H.~Lasi, P.~Fettke, H.-G. Kemper, T.~Feld, and M.~Hoffmann, ``Industry 4.0,''
  \emph{Business \& information systems engineering}, vol.~6, no.~4, pp.
  239--242, Jun. 2014.

\bibitem{ZQ_SystemsJournal_15}
Y.~Zhang, M.~Qiu, C.-W. Tsai, M.~M. Hassan, and A.~Alamri, ``Health-{CPS}:
  Healthcare cyber-physical system assisted by cloud and big data,'' \emph{IEEE
  Systems Journal}, vol.~11, no.~1, pp. 88--95, Aug. 2015.

\bibitem{HR_IoT_17}
R.~Hunzinger, ``{SCADA} fundamentals and applications in the {I}o{T},''
  \emph{Internet of things and data analytics handbook}, pp. 283--293, Dec.
  2017.

\bibitem{RH_ICT_11}
R.~Baheti and H.~Gill, ``Cyber-physical systems,'' \emph{The impact of control
  technology}, vol.~12, no.~1, pp. 161--166, 2011.

\bibitem{LY_TISSEC_11}
Y.~Liu, P.~Ning, and M.~K. Reiter, ``False data injection attacks against state
  estimation in electric power grids,'' \emph{ACM Trans. Info. Syst. Sec},
  vol.~14, no.~1, pp. 1--33, May 2011.

\bibitem{HL_IEEEIoT_17}
A.~Humayed, J.~Lin, F.~Li, and B.~Luo, ``Cyber-physical systems security —
  {A} survey,'' \emph{IEEE Internet of Things Journal}, vol.~4, no.~6, pp.
  1802--1831, May 2017.

\bibitem{Advisor_US_24}
\BIBentryALTinterwordspacing
U.~D. of~Homeland~Security. Cyber physical systems security. [Online].
  Available: \url{https://www.dhs.gov/archive/science-and-technology/cpssec}
\BIBentrySTDinterwordspacing

\bibitem{ZJ_WonSE_15}
X.~Zheng and C.~Julien, ``Verification and validation in cyber physical
  systems: Research challenges and a way forward,'' in \emph{Proc. 2015
  IEEE/ACM International Workshop on Software Engineering for Smart
  Cyber-Physical Systems}, Florence, Italy, May 2015, pp. 15--18.

\bibitem{AA_book_20}
A.~A. Jahromi and D.~Kundur, ``Fundamentals of cyber-physical systems,'' in
  \emph{Cyber-Physical Systems in the Built Environment}.\hskip 1em plus 0.5em
  minus 0.4em\relax Springer, May 2020, pp. 1--13.

\bibitem{IE_TSG_16}
I.~Esnaola, S.~M. Perlaza, H.~V. Poor, and O.~Kosut, ``Maximum distortion
  attacks in electricity grids,'' \emph{IEEE Trans. Smart Grid}, vol.~7, no.~4,
  pp. 2007--2015, Jul. 2016.

\bibitem{CKKPT_SPM_12}
S.~Cui, Z.~Han, S.~Kar, T.~T. Kim, H.~V. Poor, and A.~Tajer, ``Coordinated
  data-injection attack and detection in the smart grid: A detailed look at
  enriching detection solutions,'' \emph{IEEE Signal Process. Mag}, vol.~29,
  no.~5, pp. 106--115, Aug. 2012.

\bibitem{EPP_gsip_14}
I.~Esnaola, S.~M. Perlaza, and H.~V. Poor, ``Equilibria in data injection
  attacks,'' in \emph{Proc. IEEE Global Conference on Signal and Information
  Processing}, Atlanta, GA, USA, Dec. 2014, pp. 779--783.

\bibitem{KP_TSG_11}
T.~T. Kim and H.~V. Poor, ``Strategic protection against data injection attacks
  on power grids,'' \emph{IEEE Trans. Smart Grid}, vol.~2, no.~2, pp. 326--333,
  Jun. 2011.

\bibitem{OM_MIT_94}
M.~J. Osborne and A.~Rubinstein, \emph{A course in game theory}.\hskip 1em plus
  0.5em minus 0.4em\relax MIT press, 1994.

\bibitem{EB_dynamicGames_19}
S.~R. Etesami and T.~Ba{\c{s}}ar, ``Dynamic games in cyber-physical security:
  An overview,'' \emph{Dynamic Games and Applications}, vol.~9, no.~4, pp.
  884--913, 2019.

\bibitem{OA_ComInd_17}
H.~Orojloo and M.~A. Azgomi, ``A game-theoretic approach to model and quantify
  the security of cyber-physical systems,'' \emph{Computers in Industry},
  vol.~88, pp. 44--57, 2017.

\bibitem{MD_TSG_22}
D.~Mukherjee, ``Data-driven false data injection attack: A low-rank approach,''
  \emph{IEEE Trans. on Smart Grid}, vol.~13, no.~3, pp. 2479--2482, 2022.

\bibitem{TW_CS_22}
J.~Tian, B.~Wang, J.~Li, and C.~Konstantinou, ``Datadriven false data injection
  attacks against cyber-physical power systems,'' \emph{Computers \& Security},
  vol. 121, p. 102836, 2022.

\bibitem{Ali_SAandMSP_14}
A.~Tajer, ``Energy grid state estimation under random and structured bad
  data,'' in \emph{2014 IEEE 8th Sensor Array and Multichannel Signal
  Processing Workshop (SAM)}.\hskip 1em plus 0.5em minus 0.4em\relax IEEE,
  2014, pp. 65--68.

\bibitem{OK_TSG_11}
O.~Kosut, L.~Jia, R.~J. Thomas, and L.~Tong, ``Malicious data attacks on the
  smart grid,'' \emph{IEEE Trans. Smart Grid}, vol.~2, no.~4, pp. 645--658,
  Dec. 2011.

\bibitem{SE_SGC_17}
K.~Sun, I.~Esnaola, S.~M. Perlaza, and H.~V. Poor, ``Information-theoretic
  attacks in the smart grid,'' in \emph{Proc. IEEE Int. Conf. on Smart Grid
  Comm.}, Dresden, Germany, Oct. 2017, pp. 455--460.

\bibitem{SE_TSG_19}
------, ``Stealth attacks on the smart grid,'' \emph{IEEE Trans. Smart Grid},
  vol.~11, no.~2, pp. 1276--1285, Aug. 2019.

\bibitem{YE_SGC_20}
X.~Ye, I.~Esnaola, S.~M. Perlaza, and R.~F. Harrison, ``Information theoretic
  data injection attacks with sparsity constraints,'' in \emph{Proc. IEEE Int.
  Conf. on Smart Grid Comm.}, Tempe, AZ, USA, Oct. 2020, pp. 1--6.

\bibitem{JN_LRT_33}
J.~Neyman and E.~S. Pearson, ``On the problem of the most efficient tests of
  statistical hypotheses,'' \emph{Philosophical Trans. of the Royal Society of
  London}, vol. 231, pp. 289--337, Feb. 1933.

\bibitem{OM_TNNLS_16}
M.~Ozay, I.~Esnaola, F.~T. Yarman~Vural, S.~R. Kulkarni, and H.~V. Poor,
  ``Machine learning methods for attack detection in the smart grid,''
  \emph{IEEE Trans. Neural Netw. Learn. Syst.}, vol.~27, no.~8, pp. 1773--1786,
  Aug. 2016.

\bibitem{TM_ElementsofIT}
T.~M. Cover and J.~A. Thomas, \emph{Elements of information theory}.\hskip 1em
  plus 0.5em minus 0.4em\relax John Wiley \& Sons, Nov. 2012.

\bibitem{YE_TSG_21}
X.~Ye, I.~Esnaola, S.~M. Perlaza, and R.~F. Harrison, ``Stealth data injection
  attacks with sparsity constraints,'' \emph{arXiv preprint arXiv:2201.00065},
  2021.

\bibitem{GJ_PSanalysis_1994}
J.~J. Grainger and W.~D. Stevenson, \emph{Power system analysis}.\hskip 1em
  plus 0.5em minus 0.4em\relax McGraw-Hill, 1994.

\bibitem{AA_PSstateestimation_04}
A.~Abur and A.~G. Exposito, \emph{Power system state estimation: Theory and
  implementation}.\hskip 1em plus 0.5em minus 0.4em\relax CRC press, Mar. 2004.

\bibitem{SE_TAC_17}
D.~Shi, R.~J. Elliott, and T.~Chen, ``On finite-state stochastic modeling and
  secure estimation of cyber-physical systems,'' \emph{IEEE Trans. on Autom.
  Control}, vol.~62, no.~1, pp. 65--80, 2016.

\bibitem{GE_spawc_16}
C.~Genes, I.~Esnaola, S.~M. Perlaza, L.~F. Ochoa, and D.~Coca, ``Recovering
  missing data via matrix completion in electricity distribution systems,'' in
  \emph{Proc. Int. Workshop on Signal Processing Advances in Wireless
  Communications}, Edinburgh, United Kingdom, Jul. 2016, pp. 1--6.

\bibitem{SI_TIT_13}
I.~Shomorony and A.~S. Avestimehr, ``Worst-case additive noise in wireless
  networks,'' \emph{IEEE Trans. Inf. Theory}, vol.~59, no.~6, pp. 3833--3847,
  Jun. 2013.

\bibitem{InriaRR9481}
X.~Ye, S.~M. Perlaza, I.~Esnaola, and R.~F. Harrison, ``Stealth data injection
  attacks with sparsity constraints,'' Inria, Centre de Recherche de Sophia
  Antipolis M\'{e}dit\'{e}rran\'{e}e, Sophia Antipolis, Tech. Rep. RR-9481,
  Sep. 2022.

\bibitem{potentialgame}
D.~Monderer and L.~S. Shapley, ``Potential games,'' \emph{Games and economic
  behavior}, vol.~14, no.~1, pp. 124--143, 1996.

\bibitem{boyd_cvx}
S.~Boyd, S.~P. Boyd, and L.~Vandenberghe, \emph{Convex optimization}.\hskip 1em
  plus 0.5em minus 0.4em\relax Cambridge university press, 2004.

\bibitem{UoW_ITC_99}
\BIBentryALTinterwordspacing
U.~of~Washington. (1999) Power systems test case archive. [Online]. Available:
  \url{https://sentinel.esa.int/web/sentinel/user-guides/sentinel-2-msi/resolutions/radiometric}
\BIBentrySTDinterwordspacing

\bibitem{matpower}
R.~D. Zimmerman, C.~E. Murillo-S{\'a}nchez, and R.~J. Thomas, ``Matpower:
  Steady-state operations, planning, and analysis tools for power systems
  research and education,'' \emph{IEEE Trans. Power Syst}, vol.~26, no.~1, pp.
  12--19, Feb. 2010.

\bibitem{YE_IETSG_22}
X.~Ye, I.~Esnaola, S.~M. Perlaza, and R.~F. Harrison, ``An information
  theoretic metric for measurement vulnerability to data integrity attacks on
  smart grids,'' \emph{IET Smart Grid}, vol.~7, no.~5, pp. 583--592, 2024.

\bibitem{seber}
G.~A. Seber, \emph{A matrix handbook for statisticians}.\hskip 1em plus 0.5em
  minus 0.4em\relax John Wiley \& Sons, 2008, vol.~15.

\end{thebibliography}
